\documentclass[runningheads]{llncs}
\usepackage[T1]{fontenc}
\usepackage{graphicx}
\usepackage{hyperref}
\usepackage{color}

\usepackage[utf8]{inputenc}
\usepackage{mathtools} 
\usepackage{stmaryrd} 
\usepackage{lineno}
\usepackage{pzccal}
\usepackage{amsfonts} 
\usepackage{enumitem}\setlist[itemize]{noitemsep,topsep=-\parskip}
\usepackage{xspace}

\newcommand\hide[1]{}
\newcommand\TODO[1]{\textcolor{red}{\bf TODO: #1}}

\renewcommand\_{\raisebox{0.5pt}{\underline{\phantom{~}}}}

\newcommand\tab[1][3mm]{\hspace*{#1}}

\newcommand\eg{{\em e.g.}\,}
\newcommand\ie{{\em i.e.}\,}

\renewcommand\o\overline
\renewcommand\u\underline
\renewcommand\v\vec

\newcommand\lp{{\l\Pi}}
\newcommand\as{~\text{as}~}

\newcommand\rw\hookrightarrow
\newcommand\ra\rightsquigarrow
\newcommand\sle\subseteq

\renewcommand\th\vdash

\newcommand\ceq\coloneqq

\newcommand\bN{\mathbb{N}}

\DeclareMathOperator\imax{imax}
\DeclareMathOperator\dom{dom}
\DeclareMathOperator\FV{FV}

\newcommand\I[1]{\llbracket{#1}\rrbracket}
\newcommand\J[1]{\llparenthesis{#1}\rrparenthesis}
\newcommand\K[1]{\langle{#1}\rangle}

\renewcommand\a\alpha
\renewcommand\b\beta
\renewcommand\i\iota
\renewcommand\l\lambda
\newcommand\s\sigma
\newcommand\w\omega

\newcommand\D\Delta
\newcommand\G\Gamma

\newcommand\deq{\equiv_\Sigma}

\newcommand\va{\vec\a}

\newcommand\vl{\vec{l}}
\newcommand\vp{\vec{p}}
\newcommand\vt{\vec{t}}
\newcommand\vu{\vec{u}}
\newcommand\vx{\vec{x}}
\newcommand\vA{\vec{A}}
\newcommand\vB{\vec{B}}

\newcommand\tN{\mathtt{N}}
\newcommand\tz{\mathtt{0}}
\renewcommand\ts{\mathtt{s}}
\newcommand\tadd{\mathtt{add}}
\newcommand\tVec{\mathtt{Vec}}
\newcommand\tnil{\mathtt{nil}}
\newcommand\tcons{\mathtt{cons}}
\newcommand\Eq{\mathtt{Eq}}
\newcommand\refl{\mathtt{refl}}
\newcommand\cast{\mathtt{cast}}

\newcommand\cA{\mathcal{A}}
\newcommand\cC{\mathcal{C}}

\newcommand\cL{\mathcal{L}}
\newcommand\cM{\mathcal{M}}
\newcommand\cP{\mathcal{P}}
\newcommand\cR{\mathcal{R}}
\newcommand\cS{\mathcal{S}}
\newcommand\cT{\mathcal{T}}
\newcommand\cU{\mathcal{U}}
\newcommand\cX{\mathcal{X}}

\newcommand\lzero{\mathpzc{0}}
\newcommand\lsucc{\mathpzc{s}}
\newcommand\lmax{\mathpzc{max}}
\newcommand\limax{\mathpzc{imax}}
\newcommand\lMax{\mathpzc{Max}}
\newcommand\lC{\mathpzc{C}}
\newcommand\lV{\mathpzc{V}}

\newcommand\theory[3]{%
\vspace*{1mm}\noindent
\textcolor{blue}{\bf Theory $\Sigma_{#1}$ for #2: }%
$#3$
}

\newcommand\id{\mathsf{id}}

\newcommand\propext{\mathsf{propext}}
\newcommand\Nonempty{\mathsf{Nonempty}}
\newcommand\witness{\mathsf{witness}}
\newcommand\choice{\mathsf{choice}}

\newcommand\Quot{\mathsf{Quot}}
\newcommand\class{\mathsf{class}}
\newcommand\sound{\mathsf{sound}}
\newcommand\lift{\mathsf{lift}}

\newcommand\congrArg{\mathsf{cArg}}

\newcommand\dinst{\mathsf{inst}}
\renewcommand\v{\mathsf{v}}

\newcommand\U{\mathsf{U}}
\newcommand\Type{\mathsf{Type}}
\newcommand\Kind{\mathsf{Kind}}
\newcommand\Prop{\mathsf{Prop}}
\newcommand\llet[4]{\mathsf{let}\ #1:#2\ \ceq\ #3\ \mathsf{in}\ #4}
\newcommand\mk{\mathsf{mk}}
\newcommand\prj{\mathsf{prj}}
\newcommand\rec{\mathsf{rec}}

\newcommand\T{\mathsf{T}}
\newcommand\code[1]{{\raisebox{-0.8mm}{$\stackrel{#1}\,$}}}

\newcommand\C{\mathsf{C}}
\newcommand\eq{\mathsf{eq}}
\newcommand\gt{\mathsf{gt}}
\newcommand\lt{\mathsf{lt}}
\renewcommand\case{\mathsf{case}}

\newcommand\B{\mathsf{B}}
\newcommand\et\wedge
\newcommand\si{\mathsf{if}}

\newcommand\N{\mathsf{N}}
\newcommand\z{\mathsf{0}}
\renewcommand\s{\mathsf{s}}
\newcommand\cmp{\mathsf{cmp}}

\newcommand\set{\mathsf{set}}
\newcommand\nil{\mathsf{nil}}
\newcommand\cons{\mathsf{cons}}
\newcommand\cmpset{\mathsf{cmp}_\set}
\newcommand\add{\mathsf{add}}
\newcommand\union{\mathsf{union}}
\newcommand\incl{\mathsf{incl}}
\newcommand\fold{\mathsf{fold}}

\newcommand\typG{\mathsf{G}}

\renewcommand\S{\mathsf{S}}
\newcommand\guard{\mathsf{guard}}
\newcommand\addguard{\mathsf{addguard}}
\newcommand\incr{\mathsf{incr}}
\newcommand\cst{\mathsf{cst}}
\newcommand\var{\mathsf{var}}

\newcommand\E{\mathsf{E}}
\newcommand\Incr{\mathsf{Incr}}
\newcommand\Add{\mathsf{Add}}

\renewcommand\L{\mathsf{L}}
\newcommand\dzero{\mathsf{zero}}
\newcommand\dsucc{\mathsf{succ}}
\newcommand\dmax{\mathsf{max}}
\newcommand\dimax{\mathsf{imax}}
\newcommand\dMax{\mathsf{Max}}

\newcommand\e\varepsilon
\newcommand\sort{\mathsf{u}}

\newcommand\Fold{\mathsf{Fold}}
\newcommand\Addguard{\mathsf{Addguard}}
\newcommand\FoldAddguard{\mathsf{FoldAddguard}}

\begin{document}
\title{Encoding Lean's Type Theory in Dedukti}
%
%
\author{\href{https://blanqui.gitlabpages.inria.fr/}{Fr\'ed\'eric Blanqui}\inst{1}\orcidID{0000-0001-7438-5554} \and
\href{http://rish987.github.io}{Rishikesh Vaishnav}\inst{1}}
\authorrunning{F. Blanqui and R. Vaishnav}
%
\institute{Universit\'e Paris-Saclay, $^1$INRIA, ENS Paris-Saclay, CNRS\\Laboratoire M\'ethodes Formelles, 4 avenue des Sciences, 91190 Gif-sur-Yvette, France}
\maketitle              
\begin{abstract}

The Lean proof assistant has a rich library of mathematical
formalizations that are interesting to users of other proof
assistants. To help with the translation of this library to other
systems, we present a theory in the Dedukti logical framework in which
one can encode Lean terms and types, and define a
typability-preserving translation from some large subset of Lean to
that Dedukti theory.


\hide{\keywords{proof systems interoperability \and
  dependent type theory \and
  logical frameworks \and
  rewriting \and
  Lean \and
  Dedukti.}}
\end{abstract}

\section{Introduction}

Over the past decades, several proof assistants have been
developed for the purpose of formalizing important results in
mathematics and program verification, each with their own communities
and proof libraries. Unfortunately, it is often difficult to share
results between those proof assistants as they significantly differ in
their syntax and underlying theories \cite{kohlhase21jar}. This leads
to a duplication of work with the same mathematical theories being
formalized in different systems or, on the contrary, some mathematical
theories being available in only one system.

Since the late 90s there has been a number of prototype tools
proposed for translating from one particular system to another, \eg
\cite{mclaughlin06ijcar,keller10itp,kaliszyk13itp,carneiro16jfr}, most
of which are no longer maintained.  However, defining a separate
translation from a proof assistant to every other possible proof
assistant is a daunting task, as each translation would have to
account for the particular deduction rules of each target system. So,
one may wonder whether it is possible to have an intermediate language
to which every proof system could be translated, and which, in turn,
could be translated to every proof system.  Such a language is called
a logical framework \cite{pfenning01chapter}, and several logical
frameworks have been proposed over the years, \eg
\cite{paulson88cade,pfenning89lics,boespflug12pxtp-dedukti,rabe13mscs},
many of them being based on the $\l\Pi$-calculus \cite{harper93jacm},
which is the simplest type system combining both simple and dependent
types, that is, types that may depend on values like the type of
arrays of a given dimension.

The Dedukti logical framework extends the $\l\Pi$-calculus by enabling
the identification of types modulo a user-defined equational theory
generated from oriented equations/rewrite rules. It has been proven
that many different logics and type systems can be faithfully and
modularly represented in the $\l\Pi$-calculus modulo rewriting
\cite{cousineau07tlca,blanqui23lmcs} (rewrite rules allow for more
shallow encodings), and several tools have been developed to translate
various proof systems to Dedukti (\eg OpenTheory
\cite{assaf15pxtp-holide}, Rocq \cite{boespflug12pxtp-coqine}, Matita,
Agda \cite{genestier20fscd}, PVS \cite{hondet20types}, HOL-Light
\cite{blanqui24lpar}), and from Dedukti to various proof systems
\cite{thire18lfmtp,blanqui24lpar}.

This paper describes the first translation to Dedukti of a large
subset of Lean's type theory, called Lean$^-$, introduced in
\cite{vaishnav25ictac}. As Lean terms can be translated to Lean$^-$ in a
preliminary step using the tool
\href{https://github.com/Deducteam/Lean4Less}{\tt Lean4Less}
\cite{vaishnav25ictac}, we obtain a translation of all Lean terms to
Dedukti by composing both translations.

Because of size limitations, proofs are only sketched here but
detailed proofs can be found in \cite{blanqui26ictac-long}.

\hide{
We start by describing the type theories of the Lean proof assistant
\cite{demoura15cade,demoura21cade} and of Dedukti
\cite{saillard15phd}. We then give an overview of the translation from
Lean to Dedukti via Lean$^-$. We then provide, for each distinctive Lean
construction or feature, a Dedukti theory in which that construction
or feature can be encoded. This includes the usual constructions of
typed $\l$-calculus, algebraic universes and universe polymorphism. We
then prove that prove that the encoding of Lean$^-$ into Dedukti
preserves typing under some assumption on Lean's type theory that is
expected to hold but has not been proved yet. Finally, we present our
prototype implementation
\href{https://github.com/Deducteam/lean2dk}{\tt lean2dk} and discuss
possible improvements.
}

\subsection{Pure Type Systems}

To help compare Lean and Dedukti, we present them in the unifying
framework of Pure Type Systems (PTS)\hide{, introduced in 1989 by Berardi
and Terlouw independently as a generalization of Barendregt's cube of
type systems} \cite{barendregt92chapter}.
A PTS is specified by a tuple $\Lambda=(\cS,\cA,\cP)$ made up of a set
of sorts or universes $\cS$, a set of axioms
$\cA\sle{\cS\times\cS}$ for typing sorts, and a set of product
forming rules $\cP\sle{(\cS\times\cS)\times\cS}$. It is said
functional if $\cA$ and $\cP$ so are.
The terms and typing contexts of a PTS are:
$$t=x\mid s\mid \l x:t,t \mid t\ t\mid \Pi x:t,t\hspace*{2cm}
\G=\emptyset\mid \G,x:t$$
where $s\in\cS$, $x\in\cX$ denotes a (term) variable, $\l x\!:\!A,t$
is the function mapping a term $x$ of type $A$ to the term $t$, $(t\,u)$
is the application of the function $t$ to the term $u$, and $\Pi
x\!:\!A,B$ is the type of functions mapping a term $x$ of type $A$ to
a term of type $B$ that may depend on $x$ (written as the usual arrow
type $A\to B$ when $B$ does not depend on $x$).
We will denote by $\FV(t)$ the set of variables free in
$t$, and by $\dom(\G)$ the set of variables declared in $\G$.

A first part of PTS typing rules is recalled below:

\begin{center}
  $\cfrac{}{\th\emptyset}$~(empty)\quad
  $\cfrac{\th\G\quad \G\th A:s\in\cS\quad x\notin\G}{\th\G,x:A}$~(decl)\quad
  $\cfrac{\th\G\quad x:A\in\G}{\G\th x:A}$~(var)\\[1mm]
  
  $\cfrac{\G,x:A\th t:B\quad \G\th\Pi x:A,B:s'\in\cS}{\G\th \l x:A,t : \Pi x:A,B}$~(lam)\\[1mm]
  $\cfrac{\G\th t:\Pi x:A,B\quad \G\th u:A}{\G\th t\ u:B[x/u]}$~(app)\quad
  $\cfrac{\th\G\quad (s,s')\in\cA}{\G\th s:s'}$~(sort)\\[1mm]
  $\cfrac{\G\th A:s_A\quad \G,x:A\th B:s_B\quad ((s_A,s_B),s)\in\cP}{\G\th \Pi x:A,B:s}$~(prod)\\[1mm]
  $\cfrac{\G\th t:A\quad \G\th A\equiv B:s\in\cS}{\G\th t:B}$~(conv)
\end{center}

The last rule identifies types that are congruent modulo equations.
When those equations depend on typing constraints, like it is the case
\hide{in Martin-Löf's type theory \cite{martinlof82chapter} or} in
Lean, we use the ``semantical'' presentation of PTSs
\cite{geuvers93phd}, also called PTSs with judgmental equality. We
omit here the usual rules for reflexivity, symmetry, transitivity and
congruence.
\hide{
\begin{center}
  $\cfrac{\G\th t:A\quad \G\th A\equiv B:s\in\cS}{\G\th t:B}$~(conv)
  \quad
  $\cfrac{\G\th t:A}{\G\th t\equiv t:A}$~(refl)\\[1mm]
  
  $\cfrac{\G\th t\equiv u:A}{\G\th u\equiv t:A}$~(sym)
  \quad
  $\cfrac{\G\th t\equiv u:A\quad \G\th u\equiv v:A}{\G\th t\equiv v:A}$~(trans)\\[1mm]

  $\cfrac{\G\th A\equiv A':s_A\in\cS\quad \G,x:A\th t\equiv t':B\quad \Pi x:A,B:s\in\cS}{\G\th\l x:A,t\equiv \l x:A',t':\Pi x:A,B}$~(cong$\l$)\\[1mm]

  $\cfrac{\G\th t\equiv t':\Pi x:A,B\quad \G\th u\equiv u':A}{\G\th t\,u \equiv t'\,u':B[x/u]}$~(cong@)\\[1mm]

  $\cfrac{\G\th A\equiv A':s_A\quad \G,x:A\th B\equiv B':s_B\quad ((s_A,s_B),s)\in\cP}{\G\th\Pi x:A,B\equiv \Pi x:A',B':s}$~(cong$\Pi$)\\[1mm]

  $\cfrac{\G\th t\equiv u:A\quad \G\th A\equiv B:s\in\cS}{\G\th t\equiv u:B}$~(cong$\equiv$)
\end{center}
}
A rule that is used by all PTSs is $\b$-equivalence ($\b$-fun).
\hide{
\begin{center}
  $\cfrac{\G,x:A\th t:B\quad \G\th\Pi x:A,B:s\in\cS\quad \G\th u:A}{\G\th {(\l x:A,t)\,u}\equiv t[x/u]:B[x/u]}$~($\b$-fun)
\end{center}
}
Lean and Dedukti include other rules that we will introduce step by step.

Given a relation on terms $\cR$, let $\cM_\cR$ be its monotone closure
(smallest relation $\cM\supseteq\cR$ such that $tu\cM t'u$, $tu\cM
tu'$, $\l x\!:\!t,u\cM\l x\!:\!t',u$, \ldots when $t\cR t'$, $u\cR
u'$), and let $\G \cC_\cR\G'$ iff $\G=\G_1,x\!:\!A,\G_2$,
$\G'=\G_1,x\!:\!B,\G_2$ and $A \cM_\cR B$.

A relation $\cR$ preserves typing if $\G\th u:A$ whenever $\G\th t:A$ and
$t\cR u$.

\hide{
Historically, the first presentation of PTSs (with ($\b$-fun) only) was
using an untyped conversion \cite{barendregt92chapter}. It can be
recovered from the above presentation by simply dropping types and
typing premises in equality judgments, leading to untyped equality
judgments of the form $\G\th t\equiv u$.
}

\hide{
Barendregt's $\l$-cube is the set of the 8 functional type systems
that we obtain with $\cS=\{\Type,\Kind\}$ and $\cA=\{(\Type,\Kind)\}$
when adding to $\cP=\{((\Type,\Type),\Type)\}$ representing the
simply-typed $\l$-calculus, either type constructors with
$\{((\Kind,\Kind),\Kind)\}$, dependent types with
$\{((\Type,\Kind),\Kind)\}$, or polymorphic terms with
$\{((\Kind,\Type),\Type)\}$.
}

Many systems use an infinite hierarchy of universes
giving the PTS $\Lambda^\infty_p$:
\begin{center}
  $\cS^\infty_p =\{\U i\mid i\in\bN\}\quad
  \cA^\infty_p =\{(\U i,\U(i+1))\mid i\in\bN\}$\\
  $\cP^\infty_p =\{((\U i,\U j),\U(\max(i,j)))\mid i,j\in\bN\}$
\end{center}
This PTS is said to be predicative since a function $A\rightarrow B$
does not live in a universe smaller than the ones of $A$ and $B$.

Lean and Rocq extends this PTS by allowing the universe $\U 0=\Prop$,
used for typing propositions, to be impredicative, giving the PTS
$\Lambda^\infty_i$ \cite{coquand86lics,luo90phd}:
\begin{center}
  $\cS^\infty_i = \cS^\infty_p\quad
  \cA^\infty_i = \cA^\infty_p\quad
  \cP^\infty_i = \{((\U i,\U j),\U(\imax(i,j)))\mid i,j\in\bN\}$\\
\end{center}
where $\imax(i,0)=0$ and $\imax(i,j+1)=\max(i,j+1)$,
allowing to build a proposition by quantifying on any type,
including $\Prop$ itself.

PTSs have later been extended by adding $\eta$-equivalence
\cite{geuvers92lics}:
\begin{center}
  $\cfrac{\G\th t:\Pi x:A,B\quad x\notin\FV(t)}{\G\th t\equiv\l x\!:\!A,t\ x:\Pi x\!:\!A,B}$~($\eta$-fun)
\end{center}
local definitions $x\!:\!A\ceq t$ \cite{severi94lfcs},
\hide{
\begin{center}
  $t=\dots\mid\llet x A v b\hspace*{2cm}\G=\dots\mid\G,x:A\ceq t$\\[3mm]
  $\cfrac{\th\G\quad \G\th A:s\in\cS\quad \G\th t:A\quad x\notin\G}{\th\G,x:A\ceq t}$~(decl')
  \quad
  $\cfrac{\th\G\quad x:A\ceq t\in\G}{\G\th x:A}$~(var')\\[1mm]

  $\cfrac{\G,x:A\ceq t\th u:B}{\G\th \llet x A t u:B[x/t]}$~(let)\quad
  $\cfrac{x:A\ceq t\in\G}{\G\th x\equiv t:A}$~(def)
\end{center}
\vspace*{1mm}\noindent
}
inductive types and recursive functions
\cite{dybjer90lf,coquand88colog,stefanova99phd}, with each inductive
type declaration coming with an induction/recursion principle (also
called recursor or eliminator) and an equation for each constructor,
used to represent proofs by induction and define functions by
recursion.
For instance, the declaration of the inductive type $\tN$ of Peano
numbers with the constructors $\tz:\tN$ and $\ts:\tN\to\tN$ (successor
function) adds:
\begin{center}
  $\rec_\tN^s:\Pi Q\!:\!\tN\to s,Q\,\tz\to(\Pi n\!:\!\tN,Q\,n\to Q\,(\ts\,n))\to\Pi n\!:\!\tN,Q\,n$\\[1mm]
  $\rec_\tN^s~Q~t~u~\tz \equiv t\quad\quad
    \rec_\tN^s~Q~t~u~(\ts~n) \equiv u~n~(\rec_\tN^s~Q~t~u~n)$
\end{center}
As it may be cumbersome to define functions through the use of
recursors only, many proof assistants, including Lean, let users
define their functions using equations that are internally converted
into recursor applications \cite{goguen06chapter,cockx16jfp}. For
instance, the equations $\tadd~\tz~y = y$ and $\tadd~(\ts~x)~y =
\ts~(\tadd~x~y)$ are translated into $\tadd\ceq\l x\!:\!\tN,\l
y\!:\!\tN,\rec_\tN^s~(\l\_\!:\!\tN,\tN)~y~(\l\_\!:\!\tN,\l
r\!:\!\tN,\ts~r)~x$.

The arguments of an inductive type which are variables in the types of
its constructors are called parameters, while the others are called
indices. For instance, in the following definition of polymorphic
vectors of a given dimension:
\begin{center}
$\tVec^s: s\to\tN\to s\quad\quad
  \tnil^s: \Pi A\!:\!s,\tVec^sA~\tz$\\
  $\tcons^s: \Pi A\!:\!s,A\to\Pi n\!:\!\tN,\tVec^sA~n\to\tVec^sA~(\ts~n)$
\end{center}
the first argument of $\tVec^s$ is a parameter while its second
argument is an index.

Another common inductively defined type is Martin-L\"of's
propositional equality predicate $\Eq^s:\Pi A\!:\!s,A\to A\to\Prop$
with unique constructor $\refl^s:\Pi A\!:\!s, \Pi
x\!:\!A,\Eq^s\ A\ x\ x$ (where $A$ and $x$ are parameters), which
generates the following induction principle and equation
\cite{martinlof73lc}:
\begin{center}
  \tab[-30mm]$\rec_\Eq^{s,s'}\!:\!\Pi A\!:\!s,\Pi x\!:\!A,\Pi Q\!:\!(\Pi y\!:\!A,\Eq^sAxy\to s')$,\\
\tab[42mm]$Qx(\refl^sAx)\to\Pi y\!:\!A,\Pi e\!:\!\Eq^sAxy\to Qye$\\
  $\rec_\Eq^{s,s'}\,A~x~Q~t~x~(\refl^s\,A~x)\equiv t$
\end{center}
It can for instance be used to define a cast operator between two
equal types\hide{ (if $s\cA s'$ for some $s'$), which plays an important role
in the translation of Lean to Lean$^-$ (see Section \ref{sec-overview})}:
$$\begin{array}{l}
    \cast^s:\Pi A\!:\!s,\Pi B\!:\!s,\Eq^{s'}s\,A\,B\to A\to B\\
    \ceq \l A\!:\!s,\l B\!:\!s,\l e\!:\!\Eq^{s'}s\,A\,B,\l x\!:\!A,\rec_\Eq^{s',s}\,s\,A\,(\l B\!:\!s,\l\_\!:\!\Eq^{s'}s\,A\,B,B)\,x\,B\,e
  \end{array}$$

A full description of Lean's inductive types and recursor
rules is given in \cite{carneiro19master}.


\subsection{Universe polymorphism}

As can be seen in the previous examples, explicit universes are not
convenient and lead to definition duplications. Hence, many systems go
further and implement different forms of universe polymorphism by
implicitly taking or explicitly allowing a universe level $i\in\bN$ to
be replaced by an abstract level expression $l\in\cL$ possibly
containing level variables $\a\in\cU$
\cite{harper91tcs,herbelin05draft,sozeau14itp}:
$$l=\a\mid \lzero \mid \lsucc\ l \mid \lmax\ l\ l \mid \limax\ l\ l$$
that are compositionally interpreted, modulo a valuation
$\xi:\cU\to\bN$ for level variables, as a number
$\I{l}_\xi\in\bN$ by taking $\I\a_\xi=\xi(\a)$ and, for every function
symbol $f$,
$\I{f\,t_1\,\dots\,t_n}_\xi=\o{f}(\I{t_1}_\xi,\dots,\I{t_n}_\xi)$,
where $\o{f}:\bN^n\to\bN$ is defined as follows: $\o\lzero=0$,
$\o\lsucc(i)=i+1$, $\o\lmax(i,j)=\max(i,j)$ and
$\o\limax(i,j)=\imax(i,j)$.
The corresponding PTS $\Lambda_U$ is:
\begin{center}
  $\cS_U\!=\!\{\U l\!\mid\!l\!\in\!\cL\}~
  \cA_U\!=\!\{(\U l,\U(\lsucc\,l))\!\mid\!l\!\in\!\cL\}~
  \cP_U\!=\!\{((\U l,\U l'),\U(\limax\,l\,l'))\!\mid\!l,l'\!\in\!\cL\}$
\end{center}
together with the following additional rule identifying sorts indexed
by levels having the same denotation under any level variable
instantiation:
\begin{center}
  $\cfrac{\th\G\quad\quad\quad l\simeq l'}{\G\th\U l\equiv\U l':\U(\lsucc\,l)}$~\text{(lvl)}\quad\quad
  $\cfrac{\forall\xi,\I{l}_\xi=\I{l'}_\xi}{l\simeq l'}$
\end{center}

In this system, typing judgments may contain globally shared level
variables that are implicitly universally quantified, which can be
interpreted as families of typing judgments/derivations in the PTS
$\Lambda^\infty_i$ when varying the valuation of level variables.
Instead, Lean enforces level variables to be explicitly quantified in
declarations and applications:
\begin{center}
  $\G\ceq\dots\mid \G,x[\va]:A\mid \G,x[\va]:A\ceq t
  \quad\quad t=\dots\mid x[\vl]$\\[3mm]
  
  $\cfrac{\th\G\quad \G\th A:s\in\cS\quad x\notin\G\quad \cU(A)\sle\{\va\}}{\th\G,x[\va]:A}$~(decl)\quad
  $\cfrac{\th\G\quad x[\va]:A\in\G}{\G\th x[\vl]:A[\va/\vl]}$~(var)\\[1mm]

  $\cfrac{\th\G\quad \G\th A:s\quad \G\th t:A\quad x\notin\G\quad \cU(A,t)\sle\{\va\}}{\th\G,x[\va]:A\ceq t}$~(decl')\\[1mm]
  
  $\cfrac{\th\G\quad x[\va]:A\ceq t\in\G}{\G\th x[\vl]:A[\va/\vl]}$~(var')
  \quad
  $\cfrac{x[\va]:A\ceq t\in\G}{\G\th x[\vl] \equiv t[\va/\vl]:A[\va/\vl]}$~(def)
\end{center}
where $\va$ are pairwise distinct level variables, $\cU(t)$ are
the level variables of $t$.

For instance, propositional equality and reflexivity are declared with
$\Eq[\a]:\Pi A\!:\!\U\a,A\to A\to\Prop$ and $\refl[\a]:\Pi
A\!:\!\U\a,\Pi x\!:\!A,\Eq[\a]Axx$.

\subsection{Lean's other constructions and definitional equalities}
\label{sec-prj}

\noindent{\bf Primitive projections.} As Rocq, for efficiency reasons,
Lean features primitive projections $t.i$ ($i\in\bN$) for terms
$t:S[\vl]\vt$ where $S[\va]:\Pi\vp\!:\!\vA,\U l$ is a so-called
struct-like inductive type, that is, an inductive type having no
indices and only one constructor
$\mk_S[\va]:\Pi\vp\!:\!\vA,\Pi\vx\!:\!\vB,S[\va]\vp$. In this case,
Lean adds:
\begin{center}
  $\cfrac{\G\th v:S[\vl]\vt}{\G\th v.i:B_i\sigma}$ (proj)
  \quad
  $\cfrac{\G\th v\equiv v':S[\vl]\vt}{\G\th v.i\equiv v'.i:B_i\sigma}$ (cong-proj)\\[1mm]

  $\cfrac{\G\th\mk_S[\vl]\vt\,\vu:S[\vl]\vt}{\G\th(\mk_S[\vl]\vt\,\vu).i\equiv u_i:B_i\sigma}$ ($\b$-proj)\\[1mm]

  $\cfrac{\G\th v:S[\vl]\vt\quad \G\th v.1\equiv u_1:B_1\sigma\quad\dots\quad\G\th v.n\equiv u_n:B_n\sigma}{\G\th v\equiv \mk_S[\vl]\vt\,\vu}$ ($\eta$-proj)
\end{center}
\noindent
where $n=|\vx|$, $1\le i\le n$ and
$\sigma=[\va/\vl,\vp/\vt,x_1/v.1,\dots,x_n/v.n]$. But except for
($\eta$-proj), all the other rules are derivable by defining $v.i$ as
$\prj_S^i[\vl]\vt v$ where $\prj_S^i[\va]\!:\!\Pi\vp\!:\!\vA,\Pi
v\!:\!S[\va]\vp, B_i\theta_{i-1}\ceq\l\vp\!:\!\vA,\rec_S[\va
  l_i]\vp(\l v:S[\va]\vp,B_i\theta_{i-1})(\l\vx\!:\!\vB,x_i)$,
$\theta_i=[x_1/\prj_S^1[\va]\vp v,\dots,x_i/\prj_S^i[\va]\vp v]$,
$\rec_S[\va\b]\!:\!\Pi\vp\!:\!\vA,\Pi
Q\!:\!S[\va]\vp\!\to\!\U\b,(\Pi\vx\!:\!\vB,Q(\mk_S[\va]\vp\vx))\!\to\!\Pi
v\!:\!S[\va]\vp,Qv$ and $l_i$ is such that $\G\th B_i:\U l_i$.

\vspace*{3mm}\noindent{\bf Lean's axioms.} Lean's formalizations are
done within a non-empty context which declares several important
constants and axioms like equality, conjunction, propositional
extensionality, the axiom of choice, and the following constants for
handling quotient types \cite{carneiro19master}:

\hide{
Lean's formalizations are
done within the following builtin context which states several
important axioms like propositional extensionality and choice, and
define the notion of quotient:

\vspace*{1mm}
$\begin{array}{r@{~}c@{~}l}
  \Prop &:& \U(\lsucc\tz)\ceq\U\tz\\[1mm]
  
  \Eq[\a] &:& \Pi A\!:\!\U\a,A\to A\to\Prop\\
  \refl[\a] &:& \Pi A\!:\!\U\a,\Pi x:A,\Eq[\a]Axx\\
  \rec_\Eq[\a,\b] &:& \Pi A\!:\!\U\a,\Pi x\!:\!A,\Pi Q\!:(\Pi y\!:\!A,\Eq[\a]xy\to\U\b),\\
  && Qx(\refl[\a]Ax)\to\Pi y\!:\!A,\Pi e:\Eq[\a]Axy\to Qye\\[1mm]
  
  \wedge &:& \Prop\to\Prop\to\Prop\\
  \wedge_i &:& \Pi A\!:\!\Prop,\Pi B\!:\!\Prop,A\to B\to\wedge A B\\
  \rec_\wedge[\b] &:& \Pi A:\Prop,\Pi B:\Prop,\Pi Q:\U\b,(A\to B\to Q)\to\wedge AB\to Q\\[1mm]
  
  \leftrightarrow &:& \Prop\to\Prop\to\Prop\ceq\l A:\Prop,\l B:\Prop,\wedge(A\to B)(B\to A)\\[1mm]
  
  \propext &:& \Pi A\!:\!\Prop,\Pi B\!:\!\Prop,(A\leftrightarrow B)\to\Eq[\lsucc\tz]\Prop A B\\[1mm]

  \Nonempty[\a] &:& \U\a\to\U\a\\
  \witness[\a] &:& \Pi A\!:\!\U\a,A\to\Nonempty[\a]A\\
  \rec_\Nonempty[\a] &:& \Pi A\!:\!\U\a,\Pi Q\!:\!\Nonempty A\to\U\b,\\
  && (\Pi x\!:\!A,Q(\witness Ax))\to\Pi x\!:\!\Nonempty A,Qx\\[1mm]
  \choice[\a] &:& \Pi A\!:\!\U\a,\Nonempty A\to A\\[1mm]

  \Quot[\a] &:& \Pi A\!:\!\U\a,(A\to A\to\Prop)\to\U\a\\
  \class[\a] &:& \Pi A\!:\!\U\a,\Pi R\!:\!A\to A\to\Prop,A\to\Quot[\a]AR\\
  \rec_\Quot[\a,\b] &:& \Pi A\!:\!\U\a,\Pi R\!:\!A\to A\to\Prop,\Pi Q\!:\!\Quot AR\to\U\b,\\
  &&(\Pi x\!:\!A,Q(\class[\a]ARx))\to\Pi X\!:\!\Quot AR,QX\\[1mm]

  \sound[\a] &:& \Pi A\!:\!\U\a,\Pi R\!:\!A\to A\to\Prop,\\
  && \Pi x\!:\!A,\Pi y\!:\!A,Rxy\to\Eq[\a](\Quot[\a]AR)(\class[\a]ARx)(\class[\a]ARy)\\
\end{array}$

$\begin{array}{r@{~}c@{~}l}
  \lift[\a,\b] &:& \Pi A\!:\!\U\a,\Pi R\!:\!A\to A\to\Prop,\Pi B\!:\!\U\b,\Pi f\!:\!A\to B,\\
  && (\Pi x\!:\!A,\Pi y\!:\!A,Rxy\to\Eq[\b]B(fx)(fy))\to\Quot[\a]AR\to B\\[1mm]
\end{array}$
}

\vspace*{1mm}
$\begin{array}{r@{~}c@{~}l}
  \Quot[\a] &:& \Pi A\!:\!\U\a,(A\to A\to\Prop)\to\U\a\\
  \class[\a] &:& \Pi A\!:\!\U\a,\Pi R\!:\!A\to A\to\Prop,A\to\Quot[\a]AR\\
  \lift[\a,\b] &:& \Pi A\!:\!\U\a,\Pi R\!:\!A\to A\to\Prop,\Pi B\!:\!\U\b,\Pi f\!:\!A\to B,\\
  && (\Pi x\!:\!A,\Pi y\!:\!A,Rxy\to\Eq[\b]B(fx)(fy))\to\Quot[\a]AR\to B\\[1mm]
\end{array}$

\vspace*{1mm}
\noindent
Moreover, Lean extends the conversion with the following rules:
\begin{center}
  $\cfrac{\G\th\lift[l,m]ARBfh(\class[l]AR\,a):B}{\G\th\lift[l,m]ARBfh(\class[l]AR\,a)\equiv fa:B}$ ($\b$-quot)\\[2mm]

  $\cfrac{\G\th A:\Prop\quad \G\th v_1:A\quad \G\th v_2:A}{\G\th v_1\equiv v_2:A}$ (PI)\\[2mm]

  $\cfrac{S\text{ unit-like}^{\text{\ref{unit-type}}}\quad \G\th v_1:S[\vl]\vt\quad \G\th v_2:S[\vl]\vt}{\G\th v_1\equiv v_2:S[\vl]\vt}$ ($\eta$-unit)
\end{center}
where $S$ is
unit-like\footnote{\label{unit-type}\url{https://lean-lang.org/doc/reference/latest/Basic-Types/The-Unit-Type/}}
if it is an inductive type $S[\va]:\Pi\vp\!:\!\vA,\U l$
with a unique constructor
$\mk_S[\va]:\Pi\vp\!:\!\vA,\Pi\vx\!:\!\vB,S[\va]\vp$
with $\vB:\Prop$.

The rule (PI), called proof irrelevance, identifies any two proofs of
the same proposition \cite{debruijn94chapter}. The rule ($\eta$-unit)
is similar: any two terms of the same unit-like type are identified.

Lean axioms imply important logical principles like functional
extensionality and the excluded middle (classical logic).

\subsection{Dedukti}

Dedukti, also called the $\l\Pi$-calculus modulo rewriting
\cite{blanqui01phd-en,saillard15phd}, is an extension of the simplest
PTS of Barendregt's $\l$-cube (with an untyped conversion rule)
allowing dependent types, $\lp=(S_\lp,A_\lp,P_\lp)$, a.k.a. LF
\cite{harper93jacm}:
\begin{center}
$\cS_\lp = \{\Type,\Kind\}\quad
  \cA_\lp = \{(\Type,\Kind)\}$\\
  $\cP_\lp = \{((\Type,\Type),\Type),((\Type,\Kind),\Kind)\}$
\end{center}
Dedukti allows to extend typing contexts
with user-defined oriented equations also known as rewriting rules:
$$\G=\dots\mid \G,l\rw r$$
so that $l$ is of the form $fl_1\dots l_n$ with $f:A\in\G$ (rule
left-hand sides must be headed by a constant),
$\FV(r)\sle\FV(l)\cup\dom(\G)$ and, for the decidability of
pattern-matching, $l$ is a Miller pattern (every $x\in\FV(l)-\dom(\G)$
is applied to distinct bound variables) \cite{miller91jlc}.
Then, Dedukti extends the type conversion with:
$$\cfrac{\G=\G_1,l\rw r,\G_2\quad\dom(\theta)\sle\FV(l)-\dom(\G_1)}
{\G\th l\theta\equiv r\theta}~\text{(rew)}$$
For instance, in the context
$\G=\tN\!:\!\Type,\tz\!:\!\tN,\ts\!:\!\tN\!\to\!\tN,\tadd\!:\!\tN\!\to\!\tN\!\to\!\tN,+\,\tz\,y\rw
y,\tadd\,(\ts\,x)\,y\rw\ts\,(\add\,x\,y)$, we have
$\G\th\tadd\,(\ts\,\tz)\,(\ts\,\tz)\equiv\ts\,(\ts\,\tz)$.

In the following, we will use $\th_D$ to denote the typing relation in Dedukti.

A Dedukti theory is simply a context $\G$, that is, a set of symbol
declarations with their types together with a set of rewriting
rules. For a theory to well behave, it is required that the
combination of rewriting and $\b$-reduction is confluent, and that the
rewrite rules preserves typing (subject reduction property), two
properties that are undecidable in general but for which there exist
sufficient conditions that are decidable. For subject reduction, see
for instance \cite{blanqui20fscd}.

\subsection{Theoretical calculi vs implementations}

Up to now, we presented ideal calculi for both Lean and Dedukti. Their
implementations may however be incomplete wrt those ideal calculi in
the sense that they may be unable to check that a given term has
indeed a given type.

This may happen in Dedukti if the combination of $\b$-reduction and
rewriting is not confluent, or non-terminating.

In the case of Lean, any terminating type-checking algorithm is
necessarily incomplete since Lean's conversion relation is undecidable
\cite{carneiro19master,abel20lmcs}. Moreover, the conversion relation
implemented in Lean is not transitive, which implies that
$\b$-reduction does not preserve typing \cite{carneiro19master}.

This discrepancy between the implemented and ideal calculi comes from
the fact that they are combining a type inference algorithm with a
rewriting-based conversion algorithm which tries to prove that two
terms are convertible by reducing them both to the same term (modulo
proof irrelevance and $\eta$-rules).

Most rewrite rules are obtained by orienting the $\b$-like equations
above from left to right and linearizing the left hand-sides whenever
this is possible. Indeed, as explained in
\cite{blanqui05mscs,blanqui20fscd}, rewrite rules can often be
left-linearized without breaking subject-reduction as long as the
types of a term are convertible and type constructor are injective
modulo conversion (see Section \ref{sec-sound}).

To partially compensate for the lack of transitivity, there are two cases
though where the implemented rewrite rules are in fact more general:

\vspace*{1mm} {\bf Rewrite rule for recursors of struct-like inductive
  types.} Given a struct-like inductive type
$S[\va]\!:\!\Pi\vp\!:\!\vA,\U l$ with unique constructor
$\mk_S[\va]\!:\!\Pi\vp\!:\!\vA,\Pi\vx\!:\!\vB,S[\va]\vp$, its recursor
$$\rec_S[\va\b]\!:\!\Pi\vp\!:\!\vA,\Pi
Q\!:\!S[\va]\vp\to\U\b,(\Pi\vx\!:\!\vB,Q(\mk_S[\va]\vp\vx))\to\Pi v\!:\!S[\va]\vp,Qv$$
satisfies the equation
$\rec_S[\va\b]\,\vp\,Q\,h\,(\mk_S[\va]\vp\vx)\equiv h\vx$. But Lean
instead uses the following rewrite
rule\footnote{\href{https://github.com/digama0/lean4lean/blob/9e068dc21e6bebb8642e00c91305d03794355c46/Lean4Lean/Inductive/Reduce.lean\#L72}{https://github.com/digama0/lean4lean/blob/9e068dc21e6bebb8642e00c91305d0379\\4355c46/Lean4Lean/Inductive/Reduce.lean\#L72}}
where $n=|\vx|$ whose validity is implied by ($\eta$-proj):
$$(\b\text{-}\rec\text{-}S)\quad\rec_S[\va\b]\,\vp\,Q\,h\,v\rw h\,v.1\,\dots\,v.n$$

\vspace*{1mm}
{\bf Rewrite rule for recursors of K-like inductive types.} For
recursors on some inductively defined propositions like $\Eq$, Lean
uses a slightly more general rewrite rule, called K-like
reduction\footnote{\href{https://ammkrn.github.io/type_checking_in_lean4/type_checking/reduction.html}{https://ammkrn.github.io/type\_checking\_in\_lean4/type\_checking/reduction.html}},
which is implied by proof irrelevance. A type is a K-like type if it
is an inductive type $S[\va]:\Pi\vp\!:\!\vA,\Pi\vx:\vB,\Prop$ with a
unique constructor $\mk_S[\va]:\Pi\vp\!:\!\vA,S[\va]\vp\,\vt$ taking
no other arguments than the parameters of the inductive type. For
instance, $\Eq$ is a K-like type. Its recursor
\begin{center}
  $\begin{array}{rl}
    \rec_S[\va\b]: & \Pi\vp\!:\!\vA,\Pi Q\!:\!(\Pi\vx\!:\!\vB,S[\va]\vp\vx\to\U\b),Q\,\vt\,(\mk_S[\va]\vp)\\
    & \to\Pi\vx\!:\!\vB,\Pi v\!:\!S[\va]\vp\vx,Q\,\vx\,v\\
  \end{array}$
\end{center}
satisfies by default the equation
$\rec_S[\va\b]\,\vp\,Q\,h\,\vt\,(\mk_S[\va]\vp)\equiv h$. However
Lean uses the following more general rewrite rule instead:
$$(\b\text{-}\rec\text{-K})\quad\rec_S[\va\b]~\vp~Q~h~\vt~v\rw h\quad\text{ if }v:S[\va]\vp\vt$$
which makes the system non-terminating when $\Prop$ is impredicative
and enjoys propositional extensionality as it is the case in Lean
\cite{abel20lmcs}.


\section{Translation overview}
\label{sec-overview}

As we shall see, all the equations satisfied by Lean can be turned
into Dedukti rewrite rules, except ($\eta$-fun), (PI), ($\eta$-unit)
and ($\beta$-rec-K), which rely on typing constraints or whose
left-hand sides are not headed by a function symbol.

It is possible to encode ($\eta$-fun) by using standard rewrite rules
though \cite{goguen05popl,genestier20fscd}, but it is not very
efficient. So most Dedukti implementations support ($\eta$-fun)
primitively. On the other hand, (PI) and ($\eta$-unit) cannot be
directly encoded by using standard rewrite rules
\cite[p. 98]{vaishnav26phd}.

In \cite{vaishnav25ictac}, the second author defined Lean$^-$ as the
theory obtained by removing ($\beta$-rec-K), (PI) and ($\eta$-unit),
showed that all uses of these rules in Lean can be replaced by
applications of the cast operator defined above to derivable
propositional equalities, and implemented that translation in
\href{https://github.com/Deducteam/Lean4Less}{\tt Lean4Less}.
This translation from Lean to Lean$^-$ is a particular case of the more
general translation from extensional (ETT) to intensional type theory
(ITT) extended with functional extensionality and unicity of identity
proofs \cite{hofmann95phd,oury05tphol,winterhalter19cpp}.

\hide{(ETT) to ITT$^+$, the extension of
intensional type theory (ITT) with two axioms: functional
extensionality\hide{\footnote{$\Pi A\!\!:\!\!\U\a,\Pi
B\!\!:\!\!A\!\to\!\U\a,\Pi f\!\!:\!\!(\Pi x\!\!:\!\!A,Bx),\Pi
g\!\!:\!\!(\Pi x\!\!:\!\!A,Bx),(\Pi
x\!\!:\!\!A,\Eq[\a](Bx)(fx)(gx))\!\to\!\Eq[\a](\Pi
x\!\!:\!\!A,Bx)fg$.}} and unicity of identity proofs\hide{\footnote{$\Pi
A\!:\:\U\a,\Pi x\!:\:A,\Pi y\!:\:A,\Pi v_1\!:\:\Eq[\a]Axy,\Pi
v_2\!:\:\Eq[\a]Axy,\Eq[\tz](\Eq[\a]Axy)v_1v_2$}} (which is the
restriction of (PI) to equality)
\cite{hofmann95phd,oury05tphol,winterhalter19cpp}, where ETT is the
extension of ITT with the reflexion rule
$$\cfrac{\G\th \_:\Eq[l]Av_1v_2}{\G\th v_1\equiv v_2:A}$$
meaning that any two terms that are propositionaly equal are also
definitionaly equal (making type checking undecidable).
}

For instance, in the context
$\G=A:\Prop,v_1:A,v_2:A,B:A\to\Prop,x:Bv_1$, we have $Bv_2$ inhabited
by $x$ thanks to (PI), but $Bv_2$ is also inhabited without using (PI)
by taking instead:
$$\cast[\tz]\,(Bv_1)\,(Bv_2)\,(\congrArg[\a,\!\tz]\,A\,\Prop\,B\,(\pi\,A\,v_1\,v_2))\,x\quad\text{where}$$
$\congrArg[\a\b]\!:\!\Pi\!A\!:\!\U\a,\!\Pi\!B\!:\!\U\b,\!\Pi\!f\!:\!A\!\to\!B,\!\Pi\!x\!:\!A,\!\Pi\!y\!:\!A,\!\Eq[\a]Axy\!\to\!\Eq[\b]B(fx)(fy)$ and
$\pi:\Pi A\!:\!\Prop,\Pi v_1\!:\!A,\Pi v_2\!:\!A,\Eq[\tz]\,A\,v_1\,v_2$
are terms that can be defined without using (PI): $\congrArg$ can be
defined by using $\rec_\Eq$ only, and $\pi$, whose type is the
propositional version of (PI), \ie with $\equiv$ replaced by $\Eq$,
can be obtained by using propositional extensionality which implies
that any provable proposition is propositionaly equal to the true
proposition \cite[p. 18]{vaishnav26phd}.

Similarly, the propositional version of ($\eta$-unit) for a unit-like
type $S[\va]:\Pi\vp\!:\!\vA,\U l$ with unique constructor
$\mk_S[\va]:\Pi\vp\!:\!\vA,S[\va]\vp$, namely:
$$\Pi\vp\!:\!\vA,\Pi v_1\!:\!S[\vl]\vp,\Pi v_2\!:\!S[\vl]\vp,\Eq[l](S[\vl]\vp)\,v_1\,v_2$$
can be proved by using $\rec_S$ only (by definition of $\rec_S$, a
property $P$ holds for any element of $S$ if it holds for its unique
constructor $\mk$, and $\mk[\vl]\vt=\mk[\vl]\vt$ by reflexivity)
\cite[p. 116]{vaishnav26phd}.

Hence, our overall strategy to translate Lean to Dedukti consists of
two consecutive translation steps.
The first step is the translation from Lean to Lean$^-$ described in
\cite{vaishnav25ictac} and implemented in
\href{https://github.com/Deducteam/Lean4Less}{\tt Lean4Less}.
The second step, which is the topic of the present paper, consists in
a translation from Lean$^-$ to Dedukti which we denote by $|\cdot|$ for
terms and contexts, and $||\cdot||$ for types.

When considering a translation from a type system $S_1$ to a type
system $S_2$, we usually expect the following two properties to
hold. First, that the translation preserves inhabitation:
$(\exists t_1,\G\th_1 t_1:A_1)\implies(\exists t_2,|\G|\th_2 t_2:||A_1||)$.
If $A_1$ in inhabited in $S_1$, then $||A||$ in inhabited in $S_2$,
that is, every thing that is provable in $S_1$ is also provable in
$S_2$, or $S_2$ is powerful enough to check $S_1$-propositions, in
which case we say that the translation is sound.

Second, that the translation reflects inhabitation (the converse of
soundness):
$(\exists t_2,|\G|\th_2 t_2:||A_1||)\implies(\exists t_1,\G\th_1 t_1:A_1)$.
If $||A_1||$ is inhabited in $S_2$, then $A_1$ is also inhabited in
$S_1$, that is, $S_2$ does not prove more $S_1$-propositions than
$S_1$, in which case we say that the translation is conservative (or
complete).

Winterhalter {\em et al\,} formally proved in Rocq that the
translation from ETT to ITT is sound and conservative\hide{(ITT$^+$
does not prove more ETT-propositions than ETT)}
\cite{winterhalter19cpp}\hide{\footnote{\cite{winterhalter19cpp} proves also
that the identity translation from ITT$^+$ to ETT (for ITT$^+$ is
included in ETT) is conservative: ETT does not prove more
ITT$^+$-propositions than ITT$^+$.}}. Therefore, the first step
translation, from Lean to Lean$^-$, is sound and conservative.

In this paper, we prove that the second step translation, from
Lean$^-$ to Dedukti, is sound. Conservativity should be derivable by
extending previous results. Indeed, Cousineau and Dowek proved that
their encoding of functional PTSs is injective (modulo
$\b$-equivalence) and surjective on Dedukti terms in normal form
\cite{cousineau07tlca,cousineau09phd}, hence that their translation is
conservative if the Dedukti theory used for the encoding
terminates. It is however possible to get conservativity even in the
absence of termination by using a deep encoding of $\l$-abstraction
and application \cite{felicissimo22fscd}.

\hide{
  1) Translation from ETT to ITT

  Soundness: if G |-x t:A then |G| |- |t|:|A|

  Proof: build |G|, a valid translation of G recursively, and apply Theorem 4.4
  |\empty|=\empty
  |G,x:A|=|G|,x:|A| where |A| is obtained by Theorem 4.4 from G|-A:s

  Conservativity: if G |-x A:s and |G| |- u:|A|, then G |-x u:A

  Proof: |G| ~x G and |A| ~x A
  
  2) Translation from ITT to ETT = identity (ITT is included in ETT)

  Soundness: trivial

  Conservativity: if G |- A:s and G |-x u:A, then there is t such that G |- t:A

  Proof: G is a valid translation of G, and A is a valid translation of A.
  Therefore, by Theorem 4.4, there is t such that G |- t:A.
}


\section{Dedukti theory for PTS constructions}

For encoding the PTS constructions of Lean, we follow the pioneering
work of Cousineau and Dowek \cite{cousineau07tlca}, who proved that
any functional PTS can be encoded in Dedukti. This may seem surprising
at first glance because functional PTSs include polymorphic systems
like system F while Dedukti does not support quantifications on
types. The trick is that quantification on types can be simulated by
quantifying on type codes that can be interpreted as types thanks to
type-level rewrite rules, following the way universes can be encoded
in dependent type theory \cite[p. 87]{martinlof84book}. Another way to
see this is that the encoding is mixing both a deep encoding of PTS
types as Dedukti objects (the type codes) $|\cdot|$, and a shallow
encoding of PTS types as Dedukti types $||\cdot||$.

We first assume given a Dedukti theory $\Sigma_U$ providing a type
$\L$, and a function $|\cdot|$ translating the terms of $\cL$ into
Dedukti terms of type $\L$ (the definition of $\Sigma_U$ and $|\cdot|$
will be given in the next section). We then introduce an $\L$-indexed
type $\U$ for the deep encoding of Lean types with the constructors
$\sort$ and $\pi$ for sorts and dependent products respectively,
together with a function $\e$ interpreting those constructors in
Dedukti types (shallow encoding):

\theory{P}{PTSs}{
\U\!:\!\L\!\to\!\Type,\\
\sort\!:\!\Pi\a\!:\!\L,\U(\dsucc\,\a),~
\pi\!:\!\Pi\a\!:\!\L,\Pi\b\!:\!\L,\Pi t\!:\!\U\a,(\e\,\a\,t\!\to\!\U\b)\!\to\!\U(\dimax\,\a\,\b),\\
\e\!:\!\Pi\a\!:\!\L,\U\a\!\to\!\Type,~
\e\,\_\,(\sort\a)\rw \U\a,~
\e\,\_\,(\pi\,\a\,\b\,t\,t')\rw\Pi x\!:\!\e\,\a\,t,\e\,\b\,(t'\,x)
}

For defining the translation of a term $A$ occurring on the right
hand-side of a typing judgment, we also need to know its type, which
must be a sort \cite{barendregt92chapter}. But such a sort is not
unique: in particular, if $\G\th A:\U l$ then, for all $l'\simeq l$,
we also have $\G\th A:\U l'$. We will therefore define our translation
wrt a function $\chi$ picking a particular level for every environment
$\G$ and term $A$ typable by a sort in $\G$. In the soundness proof,
we will see that, two different sorting functions produce equivalent
terms, hence the choice of $\chi$ does not matter.

\hide{
\begin{definition}[Sorting function]
Let $\cT$ be the set of pairs $(\G,A)$ with $\G\th A\!:\!\U l$
for some $l$. A function $\chi\!:\!\cT\!\to\!\cL$ is a sorting function if,
for all $(\G,A)\in\cT$, $\G\th A\!:\!\U\chi(\G,A)$.
\end{definition}
}

\begin{definition}[Translation of PTS constructions]\label{def-trans-pts-terms}
  Let $\cT$ be the set of pairs $(\G,A)$ with $\G\th A\!:\!\U l$
for some $l$. A function $\chi\!:\!\cT\!\to\!\cL$ is a sorting function if,
for all $(\G,A)\in\cT$, $\G\th A\!:\!\U\chi_\G(A)$.
Given a sorting function $\chi$, a context $\G$ and a term $t$ typable
in $\G$, we define $|t|_\G^\chi$ in
$\Sigma_U,\Sigma_{poly},\Sigma_P$\footnote{$\Sigma_U$ is defined in
Fig. \ref{fig-univ} on page \pageref{fig-univ}, and $\Sigma_{poly}$ in
Section \ref{sec-poly} on page \pageref{sec-poly}.}  by taking
$|x|^\chi_\G = x$, $|t\,u|^\chi_\G = |t|^\chi_\G\,|u|^\chi_\G$,
$|\lambda x\!:\!A,t|^\chi_\G = \lambda
x\!:\!||A||^\chi_\G,|t|^\chi_{\G,x:A}$, $|\U l|^\chi_\G = \sort|l|$,
$|\Pi x\!:\!A,B|^\chi_\G =
\pi~|\chi_\G(A)|~|\chi_{\G,x:A}(B)|~|A|^\chi_\G~(\l
x\!:\!||A||^\chi_\G,|B|^\chi_{\G,x:A})$, where $||A||^\chi_\G =
\e\ |\chi_\G(A)|\ |A|^\chi_\G$.
\end{definition}

\hide{
Note that we use a shallow encoding for variables, abstractions and
applications, that is, PTS variables are translated to Dedukti
variables, PTS abstractions are translated to Dedukti abstractions,
and PTS applications are translated to Dedukti applications.
}

For instance, the PTS context $\G=A\!:\!\U\lzero$ is translated to
$|\G|=A:||\U\lzero||^\chi_\emptyset$ where
$||\U\lzero||^\chi_\emptyset=\e\,|l|\,|\U\lzero|=\e\,|l|\,(\sort|\lzero|)$
and $l=\chi_\emptyset(\U\lzero)$, \eg $\lsucc\lzero$. And the PTS
identity function on $A$, $(\l x\!:\!A,x)$, is translated to the
Dedukti identity function on $||A||^\chi_\G$, $(\l
x\!:\!||A||^\chi_\G,x)$, where $||A||^\chi_\G=\e\,|m|\,A$ and
$m=\chi_\G(A)$, \eg $\lzero$.

\hide{
For translating \texttt{let}'s, there are two solutions depending on
whether the targeted Dedukti checker features primitive \texttt{let}'s
or not. If the targeted checker features primitive \texttt{let}'s like
Lambdapi\footnote{\url{https://github.com/Deducteam/lambdapi}}, we can
simply translate a \texttt{let} by a \texttt{let} by taking:
$$|\llet x A t u|_\G=\llet{x}{||A||_\G}{|t|_\G}{|u|_{\G,x:||A||_\G\ceq|t|_\G}}.$$
If the targeted checker does not feature primitive \texttt{let}'s like
Dkcheck\footnote{\url{https://github.com/Deducteam/Dedukti}}, we need
to lift local definitions to global ones \cite[p. 134]{assaf15phd}.
Note that we cannot translate $\llet x A t u$ to the translation of
the $\b$-redex $(\l x:A,u)t$ because, when checking the type of $u$ in
$\llet x A t u$, we can use the fact that $x\equiv t$, which is not
the case when checking the type of $u$ in $(\l x:A,u)t$.
}


\section{Dedukti theory for universe levels}
\label{sec-univ}

For our translation from Lean to Dedukti to be sound, we in particular
have to ensure that, if two levels are equivalent in Lean, then their
translation are equivalent in some Dedukti theory $\Sigma$:
$l\simeq l'\Rightarrow\Sigma\th_D|l|\equiv|l'|$.
But, in Dedukti, two terms can be proved equivalent only if they have
the same normal form wrt some convergent rewrite system. We are
therefore looking for a finite convergent rewrite system deciding the
equational theory on natural numbers generated by $0$, $\_+1$, $\max$
and $\imax$. This theory is decidable (it is a subset of Presburger
arithmetic) and efficient algorithms have been developed for various
subcases ({\em e.g.} with $\lzero$ and $\lsucc$ only
\cite{lueker90siam}, with $\lsucc$ and $\lmax$ only \cite{bezem22tcs})
but the point here is to decide that theory using rewriting.

When there is no level variable, the task is easy since $\lmax$ and
$\limax$ can be defined by standard rewrite rules.
It is more complicated when one has level variables because, in this
case, $\lmax$ and $\limax$ satisfy non-ground algebraic equalities
like commutativity that cannot be captured by a terminating rewrite
system.
The usual strategy is to define a notion of normal form in some
extended algebra, and prove that any level expression has an
equivalent unique normal form.

In \cite{voevodsky14draft}, Voevodsky remarked that any
level expression with variables, $\lzero$, $\lsucc$ and $\lmax$ is
equivalent to an expression of the form $\lMax\,E$ where $E$ is a set
of expressions of the form (in Géran's notation \cite{geran26csl})
$\lC\ i$ or $\lV\ \a\ i$, where $\lC$ is interpreted by the identity
function and $\lV$ by the addition, such that the elements of $E$ are
pairwise incomparable, that is, is an affine function in the
$(\bN,\max,+)$ semi-ring.
The general case with $\limax$ is more complicated but has been
recently solved by Géran by guarding $\lC$ and $\lV$ by sets of level
variables:

\begin{theorem}[\cite{geran26csl}]
  Let $\o\cL$ be the set of expressions $\lMax\,E$ where $E$ is a
  finite set of sublevels, a sublevel being $\lC\ G\ (i+1)$ or
  $\lV\ (G\cup\{\a\})\ \a\ i$, where $G$ is a finite set of level
  variables, $\a$ is a level variable, and $i\in\bN$. The terms of
  $\o\cL$ are compositionally interpreted in $\bN$ by taking:
  \begin{itemize}
  \item $\o\lMax(E)=\text{if}~E=\emptyset~\text{then}~0~\text{else}~\max(E)$,
  \item $\o\lV(G,m,i)=\text{if}~0\in G~\text{then}~0~\text{else}~m+i$,
  \item $\o\lC(G,i)=\text{if}~0\in G~\text{then}~0~\text{else}~i$,
  \item $\I\emptyset_\xi=\emptyset$,
  $\I{\{\a\}}_\xi=\{\xi(\a)\}$,
  $\I{G\cup H}_\xi=\I{G}_\xi\cup\I{H}_\xi$ if $G\cap H=\emptyset$.
  \end{itemize}
  \noindent
  For all $l\in\cL$, there is a unique finite set $E$ such that
  $l\simeq\lMax\,E$ where $E$ is made of sublevels which are pairwise
  incomparable wrt the quasi-ordering $\w\preceq\w'$ such that, for
  all $\xi$, $\I\w_\xi\le\I{\w'}_\xi$.
\end{theorem}

For instance, the Géran normal form (GNF) of a variable $\a$
is $\lMax\{\lC\{\a\}\,\a\,0\}$.

One difficulty with these notions of normal form is that they are
defined up to set equality (with two kinds of sets: sets of level
variables and sets of sublevels).

\noindent
A solution is to extend Dedukti with associative-commutative
(AC) symbols \cite{ferey21phd}:
\begin{center}
  $\begin{array}{r@{\,}l}\G=\dots&\mid \G,\text{AC}(x)\\
    &(x\in\dom(\G))\end{array}$\quad
  $\cfrac{\text{AC}(x)\in\G}{\G\th xtu \equiv xut}\quad
  \cfrac{\text{AC}(x)\in\G}{\G\th x(xtu)v \equiv xt(xuv)}$
\end{center}
\noindent
and use matching modulo AC \cite{contejean04rta} to rewrite
terms. Then, a set $\{\w_1,\ldots,\w_n\}$ can be represented by
$\{\w_1\}\cup\dots\cup\{\w_n\}$ where $\cup$ is declared AC, and
duplicated elements can be removed by using the rule $x\cup x\rw x$
which can reduce $x\cup(y\cup x)$ to $x\cup y$ when using matching
modulo AC.

In \cite{genestier20fscd}, Genestier devised a finite convergent
rewrite system computing Voevodsky's normal form with this
approach. It is possible to get rid of matching modulo AC by using a
particular rewriting strategy alternating standard rewriting steps and
AC-canonization steps \cite{blanqui22fscd}. But the current
implementations of AC matching in Dedukti and of AC-canonization in
Lambdapi are not very efficient.

Instead, we propose to get rid of AC by encoding each level variable
$\a$ by a unique natural number $\i_\a$ so that the elements of a set
of level variables and of a set of sublevels can be totally
ordered. Indeed, in this case, we can represent sets of level
variables and sets of sublevels as ordered lists without duplicates
(more efficient data structures could be used instead but terms
usually contain only a few level variables). By doing so, we obtain
the Dedukti theory $\Sigma_U$ of Fig. \ref{fig-univ} that we explain
hereafter.

\begin{figure}[!p]\caption{Dedukti theory $\Sigma_U$ for universes without polymorphic constants\label{fig-univ}}

{\small
\theory{T}{basic polymorphism}{
  \T\!:\!\Type,\tau\!:\!\T\!\to\!\Type
}

\theory{C}{comparisons}{
\code\C\!:\!\T,~\C\hide{\!:\!\Type}\!\ceq\!\tau\code\C,~ \lt\!:\!\C,~ \eq\!:\!\C,~ \gt\!:\!\C,\\
\case\!:\!\Pi a\!:\!\T,\C\!\to\!\tau a\!\to\!\tau a\!\to\!\tau a\!\to\!\tau a,\\
\case~\_~\lt~x~\_~\_\rw x,~
\case~\_~\eq~\_~x~\_\rw x,~
\case~\_~\gt~\_~\_~x\rw x
}

\theory{B}{Booleans}{
\code\B\!:\!\T, ~\B\hide{\!:\!\Type}\!\ceq\!\tau\code\B,~ \!\top\!:\!\B,~ \bot\!:\!\B,\\
\et\!:\!\B\!\to\!\B\!\to\!\B,~
\!\top\,\et\,\!\top \rw \!\top,~
\bot\,\et\,\_ \rw \bot,~
\_\,\et\,\bot \rw \bot,\\
\si\!:\!\Pi a\!:\!\T,\B\!\to\!\tau a\!\to\!\tau a\!\to\!\tau a,~
\si~\!\top~x~\_ \rw x,~
\si~\bot~\_~x \rw x
}

\theory{N}{naturals}{
\code\N\!:\!\T,~ \N\hide{\!:\!\Type}\!\ceq\!\tau\code\N,~ \z\!:\!\N,~ \s\!:\!\N\!\to\!\N,~ \cmp_\N\!:\!\N\!\to\!\N\!\to\!\C,\\
\cmp_\N~\z~\z\rw\eq,~
\cmp_\N~(\s~\_)~\z\rw\gt,~
\cmp_\N~\z~(\s~\_)\rw\lt,~
\cmp_\N~(\s~x)~(\s~y)\rw\cmp_\N~x~y,\\
=_\N\hide{:\!\N\!\to\!\N\!\to\!\B}\ceq\!\l x\!:\!\N,\l y\!:\!\N,\case~\code\B~(\cmp_\N~x~y)~\!\bot~\!\top~\bot
\le_\N\hide{:\!\N\!\to\!\N\!\to\!\B}\ceq\!\l x\!:\!\N,\l y\!:\!\N,\case~\code\B~(\cmp_\N~x~y)~\!\top~\!\top~\bot
}

\theory{\set}{polymorphic sets}{\\
\set\!:\!\T\!\to\!\T,~ \nil\!:\!\Pi a\!:\!\T,\tau(\set~a),~ \cons\!:\!\Pi a\!:\!\T,\tau a\!\to\!\tau(\set~a)\!\to\!\tau(\set~a),\\[1mm]
\fold\!:\!\Pi a\!:\!\T,\Pi b\!:\!\T,(\tau a\!\to\!\tau b\!\to\!\tau b)\!\to\!\tau(\set~a)\!\to\!\tau b\!\to\!\tau b,\\
\fold~a~b~f~(\nil~\_)~y \rw y,~
\fold~a~b~f~(\cons~\_~x~S)~y \rw \fold~a~b~f~S~(f~x~y),\\[1mm]
\add\!:\!\Pi a\!:\!\T,(\tau a\!\to\!\tau a\!\to\!\C)\!\to\!\tau a\!\to\!\tau(\set~a)\!\to\!\tau(\set~a),\\
\add~a~\_~x~(\nil~\_) \rw \cons~a~x~(\nil~a),\\
\add~a~f~x~(\cons~\_~y~S\as T)
\rw \case~(\set~a)~(f~x~y)~(\cons~a~x~T)~T~(\cons~a~y~(\add~a~f~x~S)),\\[1mm]
\union\!:\!\Pi a\!:\!\T,(\tau a\!\to\!\tau a\!\to\!\C)\!\to\!\tau(\set~a)\!\to\!\tau(\set~a)\!\to\!\tau(\set~a),\\
\union~\_~\_~(\nil~\_)~T \rw T,~
\union~a~f~(\cons~\_~x~S)~T \rw \add~a~f~x~(\union~a~f~S~T),\\[1mm]
\incl\!:\!\Pi a\!:\!\T,(\tau a\!\to\!\tau a\!\to\!\C)\!\to\!\tau(\set~a)\!\to\!\tau(\set~a)\!\to\!\B,~
\incl~\_~\_~(\nil~\_)~\_ \rw \!\top,\\
\incl~\_~\_~(\cons~\_~\_~\_)~(\nil~\_) \rw \bot,\\
\incl~a~f~(\cons~\_~x~S'\as S)~(\cons~\_~y~T') \rw
\case~\code\B~(f~x~y)~\bot~(\incl~a~f~S'~T')~(\incl~a~f~S~T'),\\[1mm]
\cmpset\!:\!\Pi a\!:\!\T,(\tau a\!\to\!\tau a\!\to\!\C)\!\to\!\tau(\set~a)\!\to\!\tau(\set~a)\!\to\!\C,~
\cmpset~\_~\_~(\nil~\_)~(\nil~\_) \rw \eq,\\
\cmpset~\_~\_~(\nil~\_)~(\cons~\_~\_~\_) \rw \lt,~
\cmpset~\_~\_~(\cons~\_~\_~\_)~(\nil~\_) \rw \gt,\\
\cmpset~a~f~(\cons~\_~x~S)~(\cons~\_~y~T) \rw \case~\code\C~(f~x~y)~\lt~(\cmpset~a~f~S~T)~\gt
}

\theory{G}{sets of naturals}{
\code\typG\hide{\!:\!\T}\!\ceq\!\set\,\code\N,~\typG\hide{\!:\!\Type}\!\ceq\!\tau\code\typG,~
\cmp_\typG\hide{\!:\!\typG\!\to\!\typG\!\to\!\C}\!\ceq\!\cmpset~\code\N~\cmp_\N
}

\theory{S}{sublevels}{
\code\S\!:\!\T,~ \S\hide{\!:\!\Type}\!\ceq\!\tau\code\S,~
\var\!:\!\typG\!\to\!\N\!\to\!\N\!\to\!\S,~ \cst\!:\!\typG\!\to\!\N\!\to\!\S,\\
\cmp_\S\!:\!\S\!\to\!\S\!\to\!\C,~
\cmp_\S~(\cst~\_~\_)~(\var~\_~\_~\_)\rw\lt,~
\cmp_\S~(\var~\_~\_~\_)~(\cst~\_~\_)\rw\gt,\\
\cmp_\S~(\cst~G~i)~(\cst~H~j) \rw \case~\code\C~(\cmp_\typG~G~H)~\lt~(\cmp_\N~i~j)~\gt,\\
\cmp_\S~(\var~G~x~i)~(\var~H~y~j) \rw \case~\code\C~(\cmp_\typG~G~H)~\lt~(\case~\code\C~(\cmp_\N~x~y)~\lt~(\cmp_\N~i~j)~\gt)~\gt,\\[1mm]
\guard\!:\!\S\!\to\!\typG,~ \guard~(\var~G~\_~\_)\rw G,~
\guard~(\cst~G~\_)\rw G,\\[1mm]
\addguard\!:\!\S\!\to\!\S\!\to\!\S,~
\addguard~(\cst~G~i)~\w \rw \cst~(\union~\code\N~\cmp_\N~G~(\guard~\w))~i,\\
\addguard~(\var~G~x~i)~\w \rw \var~(\union~\code\N~\cmp_\N~G~(\guard~\w))~x~i,\\[1mm]
\incr\!:\!\S\!\to\!\S,~ \incr~(\var~G~x~i) \rw \var~G~x~(\s~i),~
\incr~(\cst~G~i) \rw \cst~G~(\s~i),\\[1mm]
\le:\!\S\!\to\!\S\!\to\!\B,~
\var~\_~\_~\_ \le \cst~\_~\_ \rw \bot,~
\cst~G~i \le \cst~H~j \rw \incl~\code\N~H~G\et i\le_\N j,\\
\cst~G~i \le \var~H~\_~j \rw \incl~\code\N~H~G\et i\le_\N\s~j,\\
\var~G~x~i \le \var~H~y~j \rw \incl~\code\N~H~G~\et~x =_\N y~\et~ i\le_\N j
}

\theory{E}{sets of sublevels}{
\E\hide{\!:\!\Type}\!\ceq\!\tau(\set\,\code\S),\\
\Incr\!:\!\E\!\to\!\E,~
\Incr~(\nil~\_) \rw \nil~\code\S,~
\Incr~(\cons~\_~\w~S) \rw \cons~\code\S~(\incr~\w)~(\Incr~S),\\[1mm]
\Add\!:\!\S\!\to\!\E\!\to\!\E,~
\Add~\w~(\nil~\_) \rw \cons~\code\S~\w~(\nil~\code\S),~
\Add~\w~(\cons~\_~\w'~E\as F)\\
\rw \si~(\set\,\code\S)~(\w\le \w')~F~(\si~(\set\,\code\S)~(\w'\le \w)~(\add~\code\S~\cmp_\S~\w~E)~(\add~\code\S~\cmp_\S~\w'~(\Add~\w~E)))
}

\theory{I}{universes with $\imax$}{\\
  \L\!:\!\Type,~
  \dzero\!:\!\L,~ \dsucc\!:\!\L\!\to\!\L,~ \dmax\!:\!\L\!\to\!\L\!\to\!\L,~ \dimax\!:\!\L\!\to\!\L\!\to\!\L,~ \dMax\!:\!\E\!\to\!\L,\\
\dzero \rw \dMax~(\nil~\code\S),~
\dsucc~(\dMax~E) \rw \dMax~(\Add~(\cst~(\nil~\code\N)~(\s~\z))~(\Incr~E)),\\
\Fold\hide{\!:\!(\S\!\to\!\E\!\to\!\E)\!\to\!\E\!\to\!\E\!\to\!\E}\!\ceq\!\fold~\code\S~(\set\,\code\S),~
\dmax~(\dMax~E)~(\dMax~F) \rw \dMax~(\Fold~\Add~E~F),\\
\Addguard\hide{\!:\!\S\!\to\!\S\!\to\!\E\!\to\!\E}\!\ceq\! \l u\!:\!\S,\l v\!:\!\S,\Add~(\addguard~u~v),\\
\FoldAddguard\hide{\!:\!\E\!\to\!\S\!\to\!\E\!\to\!\E}\!\ceq\! \l F\!:\!\E,\l u\!:\!\S,\Fold~(\Addguard~u)~F,\\  
\dimax~(\dMax~\_)~(\dMax~(\nil~\_)) \rw \dMax~(\nil~\code\S),\\
\dimax~(\dMax~E)~(\dMax~(\cons~\_~\_~\_\as F)) \rw \dMax~(\Fold~(\FoldAddguard~F)~E~F)
}
}
\end{figure}


First note that Dedukti does not allow quantifications on types. It is
however possible to simulate polymorphism by quantifying on type codes
instead (whose type is $\T$) which can be interpreted as actual types
by applying the function $\tau$. Hence, for instance, $\code\C$ is
type code of the type $\C$, and the $\case$ function can be used to
generate values of type $\tau\,a$ for any type code $a$.

The system $\Sigma_U$ is a dependently-typed higher-order rewrite
system that could be turned into a fully simply-typed first-order
rewrite system by unfolding the definition of $\fold$ and by
duplicating the definition of each polymorphic symbol for each type
code in $\T$ (monomorphisation).

The theories for comparisons, Booleans and natural numbers provide a
few standard functions on these data types. Sets are implemented by
using the function $\add$ which returns an ordered list without
duplicates if its second argument so is. $\incl$ is the inclusion
predicate. $\cmpset$ turns a total order on $\tau\,a$ into a total
order on $\tau(\set\,a)$.  $\cmp_\S$ provides a total order on
sublevels. $\guard$ returns the guard set of a
sublevel. $\addguard\,s\,s'$ extends the guard set of $s'$ with the
guard set of $s$. $\incr$ increases the constant part of a sublevel by
$1$. Finally, $\le$ implements the quasi-ordering $\preceq$
\cite{geran26csl}.

\begin{definition}\label{def-trans-level}
  Let $\Sigma_U$ be the Dedukti theory in Fig. \ref{fig-univ}. The
  function $|\cdot|_c$ from levels to Dedukti terms of type $\L$ is
  defined by $|\a|_c =
  \dMax(\cons\,\code\S\,(\var\,(\nil\,\code\N)\,\u{\i_\a}\,\z)\,(\nil\,\code\S))$,
  $|\lzero|_c=\dzero$, $|\lsucc\,l|_c=\dsucc\,|l|_c$,
  $|\lmax\,l_1\,l_2|_c=\dmax\,|l_1|_c\,|l_2|_c$, and
  $|\limax\,l_1\,l_2|_c=\dimax\,|l_1|_c\,|l_2|_c$, where $\i$ is a
  bijection from level variables to $\bN$ and, for all natural number
  $n$, $\u{n}$ is the Dedukti term of type $\N$ such that $\u{0}=\z$
  and $\u{n+1}=\s\,\u{n}$.
\end{definition}

We then check that equivalent levels are translated to equivalent terms:

\begin{theorem}
  \label{th-univ-adequate}
  $l_1\simeq l_2$ iff $\Sigma_U\th_D|l_1|_c\equiv|l_2|_c$.
\end{theorem}

\begin{proof}
We give a detailed proof in \cite{blanqui26ictac-long}. It can be
summarized as follows. The rewrite system $\Sigma_U$ is weakly
orthogonal, hence confluent, and terminating. Hence, every Dedukti
term has a unique normal form. We then extend the semantics of level
expressions to $\Sigma_U$, and prove that every rule preserves it,
using equations proved by Géran \cite{geran26csl}. We then define a
function translating Dedukti normal forms back to $\o\cL$, and prove
that it is injective. The main difficulty is to come up with the right
definition of semantics to properly handle the replacement of $\a$ by
$\o{\i_\a}$.
\end{proof}


\section{Dedukti theory for universe polymorphism}
\label{sec-poly}

In the previous section, we defined a Dedukti theory to handle the
equational theory of universe levels, assuming that level variables
are encoded as natural numbers. However, this encoding does not allow
us to handle the typing rules for declaring and applying
universe-polymorphic constants where level variables need to be
abstracted over and substituted by arbitrary levels. To handle these
rules, we propose to use an hybrid encoding mixing both a shallow
embedding, where a level variable is represented by a Dedukti variable
to be able to abstract over and substitute it, and a deep embedding,
where a level variable is represented by a natural number to be able
to totally order level variables. To this end, we introduce two new
symbols:
\begin{itemize}
\item $\v$ which takes two arguments: the level variable itself
  (shallow encoding) and a natural number uniquely representing it
  (deep encoding),
\item $\dinst$ to mark a level term as an instantiation.
\end{itemize}
Morever, we add two rewrite rules on $\v$, which must be tried in the
given order. A term of the form $\v\,l\,i$, that is the translation of
a level variable, rewrites to $l'$ with the first rule if $l$ is of
the form $\dinst\,l'$, and to the GNF of $i$ otherwise.

\theory{poly}{universe polymorphism}{\\
\v\!:\!\L\!\to\!\N\!\to\!\L,~
\dinst\!:\!\L\!\to\!\L,~
\v\,(\dinst\,l)\,\_ \rw\,l,~
\v~\_~i \rw \dMax(\cons\,\code\S\,(\var\,(\nil\,\code\N)\,i\,\z)\,(\nil\,\code\S))
}

\vspace*{1mm}
The rewrite system for $\v$ is not confluent in general. However, it
is confluent if the first rule is always tried before the second one,
which is a reduction strategy supported by all Dedukti
checkers\hide{\footnote{Both in Dkcheck and Lambdapi, rules are declared by
blocks. In Dkcheck, blocks are tried sequentially and in some
unspecified order within a block. In Lambdapi, rules are tried in some
unspecified order except when the symbol is declared {\tt sequential}
in which case rules are tried in the order they are declared.}}.

\begin{definition}[Translation of universe levels]
  \label{def-trans-level-poly}
  The function $|\cdot|$ from levels to Dedukti terms of type $\L$ in
  $\Sigma_U,\Sigma_{poly}$ is defined as in Def.
  \ref{def-trans-level} except for level variables where we take
  $|\a|=\v\,\a\,\o{\i_\a}$.
\end{definition}

We check that equivalent levels are translated to equivalent terms:

\begin{theorem}
  \label{th-univ-poly-correct}
 If $l_1\simeq l_2$ then
 $\Sigma_U,\Sigma_{poly}\th_D|l_1|\equiv|l_2|$.
\end{theorem}

\begin{proof}
  If $l_1\simeq l_2$ then, by Theorem \ref{th-univ-adequate},
  $\Sigma_U\th_D|l_1|_c\equiv|l_2|_c$. But, for all $l$,
  $\Sigma_U,\Sigma_{poly}\th_D|l|\equiv|l|_c$. Therefore,
  $\Sigma_U,\Sigma_{poly}\th_D|l_1|\equiv|l_2|$.
\end{proof}

\begin{definition}[Translation of terms]
  We extend Def. \ref{def-trans-pts-terms} by taking
  $|\G,x[\vec\a]\!:\!A| = |\G|,x\!:\!(\Pi\vec\a\!:\!\L,||A||_\G)$
  and $|x[l_1,\dots,l_n]|^\chi_\G = x\,(\dinst\,|l_1|)\,\dots\,(\dinst\,|l_n|)$.
\end{definition}

For instance, with $\G=(\id[\a]\!:\!(\Pi A\!:\!\U\a,A\!\to\!A),
A\!:\!\U\lzero, x\!:\!A)$, $t=\id[\lzero]A x$ and $\i_\a=0$, the
translation gives (modulo some rewriting):
$|\G|=(\id\!:\!(\Pi\a\!:\!\L,\Pi
A\!:\!\U(\v\a\z),\e(\v\a\z)A\to\e(\v\a\z)A), A\!:\!\U(\dzero),
x\!:\!\e(\dzero)A)$, $|t|_\G=\id(\dinst(\dzero))A x$ and
$||A||_G=\e(\dzero)A$. We have $\G\th t:A$ in Lean and
$\Sigma,|\G|\th_D|t|_\G:||A||_\G$ in Dedukti because
$\id(\dinst(\dzero))$ is of type $\Pi A\!:\!\U(\v(\dinst(\dzero))\z)$,
$\e(\v(\dinst(\dzero))\z)A\to\e(\v(\dinst(\dzero))\z)A$ and
$\v(\dinst(\dzero))\z\rw\dzero$ by the first rule of $\v$.


\section{Dedukti theory for the other Lean's constructions}

The translation is defined on well-typed terms only but, because of
linearization, the left-hand side of a rewrite rule may not be
typable. Typability can be recovered by undoing linearisation
though. Hence, we can translate every Lean rewrite rule $l\rw r$ into
the Dedukti rule $|l\rho|^\chi_\G\rho^{-1}\rw|r|^\chi_\G$ where $\rho$
denotes linearization and $\G$ is an environment in which $l\rho$ is
typable. For more details on rule linearisation, see \cite{blanqui05mscs}.

\vspace*{1mm}
{\bf Recursors.} Instead of a general definition, we show how the
translation works on polymorphic vectors.
\hide{
$$\begin{array}{rl}
  \tVec[\a]: & \U\a\to\tN\to\U\a\\
  \tnil[\a]: & \Pi A\!:\!\U\a,\tVec[\a]A\tz\\
  \tcons[\a]: & \Pi A\!:\!\U\a,A\to\Pi n\!:\!\tN,\tVec[\a]An\to\tVec[\a]A(\ts n)\\
  \rec_\tVec[\a\b]: & \Pi A\!:\!\U\a,\Pi Q\!:\!\Pi n\!:\!\tN,\tVec[\a]An\to\U\b,\\
  &Q\tz(\tnil[\a]A)\to\\
  &(\Pi x\!:\!A,\Pi n\!:\!\tN,\Pi v\!:\!\tVec[\a]An,Qnv\to Q(\ts n)(\tcons[\a]xnv))\to\\
  &\Pi n\!:\!\tN,\Pi v\!:\!\tVec[\a]An,Qnv\\
\end{array}$$
}
The equations satisfied by the recursor are:
$$\begin{array}{r@{~\equiv~}l}
  \rec_\tVec[\a\b]AQfg\,\tz\,(\tnil[\a]A) & f\\
  \rec_\tVec[\a\b]AQfg\,(\ts n)\,(\tcons[\a]Axnv) & gxnv(\rec_\tVec[\a\b]AQfg\,nv)\\
\end{array}$$
The (left-linear) rules actually used in Lean are:
$$\begin{array}{r@{~\rw~}l}
  \rec_\tVec[\a\b]AQfg\,\_\,(\tnil[\a']A') & f\\
  \rec_\tVec[\a\b]AQfg\,\_\,(\tcons[\a']A'xnv) & gxnv(\rec_\tVec[\a\b]AQfg\,nv)\\
\end{array}$$
which we translate to Dedukti as:
$$\begin{array}{r@{~\rw~}l}
  \rec_\tVec\a\b AQfg\,\_\,(\tnil\a' A') & f\\
  \rec_\tVec\a\b AQfg\,\_\,(\tcons\a' A'xnv) & gxnv(\rec_\tVec\a\b AQfg\,nv)\\
\end{array}$$

{\bf Quotients.} Similarly, the rewrite rule for quotients is
translated to:
$$\lift\,\a\,\b\,A\,R\,B\,f\,h\,(\class\,\a'\,A'\,R'\,a)\rw f\,a$$

{\bf Primitive projections} are translated to
applications of actual projection functions as follows. If
$S[\va]:\Pi\vp\!:\!\vA,\U l$ is a struct-like inductive type with
unique constructor
$\mk_S[\va]:\Pi\vp\!:\!\vA,\Pi\vx\!:\!\vB,S[\va]\vp$, and $\G\th
v:S[\vl]\vt$, then:
$$|v.i|_\G=\prj_S^i\,(\dinst\,|\vl|)\,|\vt|_\G\,|v|_\G$$
where $\prj_S^i\!:\!\Pi\va\!:\!\L,||\Pi\vp\!:\!\vA,S[\va]\vp\to
B_i\sigma||^\chi_\D$, $\D$ is the context in which $S$ is defined and
$\sigma$ is like in Sec. \ref{sec-prj}. Moreover, we add the
(left-linear) rules:
$$\prj_S^i\,\va\,\vp\,(\mk_S\,\va'\,\vp'\,\vx)\rw x_i$$
Then, we can translate the rule for recursors on struct-like
inductive types by:
$$\rec_S\,\va\,\b\,\vp\,Q\,h\,x\rw h\,(\prj_S^1\,\va\,\vp\,x)\,\dots\,(\prj_S^n\,\va\,\vp\,x)$$
Finally, to take care of ($\eta$-proj), we can add the (non-left-linear) rule\\
\hspace*{2.7cm}$\mk_S\,\va\,\vp\,(\prj_S^1\,\va^1\,\vp^1\,x)\,\dots\,(\prj_S^n\,\va^n\,\vp^n\,x)\rw x$\\
But a builtin handling like for ($\eta$-fun) would be better.


\section{Soundness of the translation}
\label{sec-sound}

The soundness proof below crucially relies on Lean$^-$ satisfying the
following two non-trivial properties:
\begin{itemize}
\item convertibility of types: $\G\th t:A$ and $\G\th t:B$ imply
  $\G\th A\equiv B:C$,
\item injectivity of $\U$: $\G\th\U l\equiv\U l':C$ implies $l\simeq
  l'$.
\end{itemize}

These properties hold in the absence of $\eta$-rules since, in this
case, the untyped rewrite relation $\hookrightarrow$ obtained by
orienting equations from left to right is orthogonal and thus
confluent \cite{klop93tcs,oostrom94lfcs}, and $\equiv$ is equal to the
joinability relation $\hookrightarrow^*{}^*\hookleftarrow$ (two terms
are equivalent iff they have a common reduct).

However, the combination of ($\b$-fun) and ($\eta$-fun) is not
confluent on untyped terms with type-annotated abstractions
\cite{nederpelt73phd}, and the combination of ($\b$-fun), ($\b$-proj)
and ($\eta$-proj) is not confluent on untyped terms
\cite{klop80phd}. Confluence of these combinations can be recovered by
using conditional rewriting \cite{devrijer89lics} or simply typed
terms \cite{pottinger81ndj}. This later result has been extended to
($\eta$-unit) and polymorphism in \cite{curien96jfp}, and to dependent
types with one universe in \cite{adjedj24cpp}. Finally, in
\cite{carneiro26draft}, convertibility of types and injectivity of
$\U$ are (partially formally) proved for a (non-terminating) type
theory with ($\beta$-fun), ($\eta$-fun), ($\beta$-proj),
($\eta$-proj), ($\eta$-unit), (PI), the inductive types $\bN$ and
$\Eq$, and a general fixpoint combinator. Some work is ongoing to
(formally) prove that the above properties hold for the whole Lean
type theory \cite{carneiro25draft} (the proof given in
\cite{carneiro19master} is incorrect [personal communication from the
  author]).

\hide{
Dedukti enjoys the untyped version of the last property if the
combination of rewrite rules and $\b$-reduction is confluent and
rewrite rules preserve typing \cite[Lemma 41]{blanqui01phd}.
}

\begin{theorem}\label{th-sound}
  For all sorting function $\chi$, if $\G\th t:A$ in Lean$^-$, then
  $\Sigma,|\G|^\chi\th_D|t|^\chi_\G:||A||^\chi_\G$, where $\Sigma$ is
  the extension of $\Sigma_U,\Sigma_{poly},\Sigma_P$ with the
  translation of Lean's axioms and user-defined inductive types and
  recursors.

  Moreover, for all other sorting function $\chi'$, we have
  $|\G|^{\chi'}\cC_\leftrightarrow^*|\G|^\chi$,
  $|t|^{\chi'}_\G\cM_\leftrightarrow^*|t|^\chi_\G$ and
  $||A||^{\chi'}_\G\cM_\leftrightarrow^*||A||^\chi_\G$, where
  $\cdot^*$ is the reflexive-transitive closure and $\leftrightarrow$
  replaces a subterm $|l|$ by $|l'|$ if $l\simeq l'$.
\end{theorem}

\begin{proof}
  We give a detailed proof in \cite{blanqui26ictac-long}. It
  essentially follows Cousineau and Dowek's proof for functional PTSs
  \cite{cousineau07tlca}, with some extra care to properly handle the
  fact that the sorts of a term are not equal anymore but equivalent.
\end{proof}


\section{Related work}

In the pioneering work at the origin of Dedukti
\cite{cousineau07tlca,cousineau09phd}, Cousineau and Dowek defined a
translation to Dedukti of functional pure type systems with universe
constants and type convertibility modulo $\equiv_\b$ (PTS/$\b$).

In \cite{boespflug12pxtp-coqine}, Boespflug and Burel developed the
first version of CoqInE, a (partial) translator from Rocq to Dedukti,
but collapsed all universes (so they actually encode the inconsistent
and non-terminating system $\l *$ with $\Type\!:\!\Type$). However
they provide an encoding of Rocq inductive types and recursive
functions.

In \cite{assaf15phd}, Assaf extended Cousineau and Dowek's work by
defining a translation to Dedukti of cumulative\footnote{Cumulativity
adds a subtyping relation on sorts.} type systems with universe
constants and type convertibility modulo $\equiv_\b$ (CTS/$\b$) having
the ``local minimum property''. Assaf also sketched an encoding of
$\Lambda^\infty_i/\b$, where universes are natural numbers and $0$ is
impredicative, extended with cumulativity (absent in Lean), and
developed Krajono, a (partial) translator from Matita
\cite{asperti11cade} to Dedukti.

In \cite{thire20phd}, Thiré remarked that Assaf' soundness proof is
actually incorrect, and proved that the translation is sound for
``normal'' CTS/$\b$.
In \cite{felicissimo24fscd}, Felicissimo and Winterhalter extended
Assaf's work on $\Lambda^\infty_i/\b$ to product covariance.

In \cite{ferey21phd}, Férey extended Assaf's work on
$\Lambda^\infty_i/\b$ with universe polymorphism and explicit universe
constraints, hence extending the type convertibility relation with an
equivalence on universe expressions ($\Lambda^\infty_i/\b\!\simeq$),
and developed the current version of CoqInE.
However, although Férey considers the sort $\Prop$ to be
impredicative, its equivalence relation on universe levels does not
include Lean's $\imax$ operator. Moreover, the soundness relies on a
number of assumptions satisfied by Rocq but only after some
pre-processing \cite{herbelin05draft}.
\hide{
Note also that, in contrast with previous works, Férey defines a
translation on triples $(\G,t,\pi)$, where $\pi$ is a typing
derivation of $t$ in $\G$, and proves that the result does not depend
on $\pi$. Here, we define a translation of $(\G,t,\chi)$ where $\chi$
is a sorting function and prove that the choice of that function does
not matter (as long as convertibility of types holds).
}

In \cite{genestier20fscd}, Genestier provided an encoding to Dedukti
of PTS/$\b$ with universe polymorphism, a convergent rewrite system
for deciding the equivalence of universe expressions with variables
and $\max$ (but not $\imax$), and a (partial) translator from Agda to
Dedukti.
In \cite{assaf16hatt}, Assaf {\em et al\,} considered the case of
universe expressions with variables and $\imax$ (called {\it rule} in
their work) but the proposed rewrite system is incomplete wrt the
equational theory of $\imax$.
A complete theory was finally given by Géran in \cite{geran26csl}.

This work differs from those previous works in several aspects. This
is the first work considering the type theory of Lean and providing a
prototype translator. This extends Genestier's work
\cite{genestier20fscd} from $\Lambda^\infty_p/\b$ to
$\Lambda^\infty_i/\b\!\simeq$. We provide a rewrite system computing
Géran's normal form for universe expressions with $\imax$
\cite{geran26csl}, and prove its correctness. And we introduce a
simple Dedukti theory for handling the instantiation of
universe-polymorphic definitions. Since Lean does not feature universe
cumulativity, which breaks convertibility of types, we can get rid of
the complexities necessary for handling cumulativity. In particular,
we can define our translation on terms like Cousineau and Dowek
(modulo a sorting function), and not on typing derivations like in the
other works.

\section{Conclusion and future work}

In \cite{vaishnav25ictac}, the second author showed how to translate
Lean to a simpler theory called Lean$^-$. Here, we defined a
translation of Lean$^-$ to a theory in the Dedukti logical framework,
that is sound for all sub-systems in which a term has a unique type
modulo convertibility, which is the topic of some recent work
\cite{carneiro26draft}.

We started to implement that translation in
\url{https://github.com/Deducteam/lean2dk} as a fork of the mirror
implementation of Lean in Lean called Lean4Lean
\cite{carneiro25draft}. While our translation shows some promising
preliminary results on the translation of specific test files and
small mathematical libraries, more work needs to be done to handle
Lean's handling of native (bignum) arithmetic operators and scale to
larger libraries.


%
%
\bibliographystyle{splncs04}
\bibliography{main}


\newpage
{\bf Proof of Theorem \ref{th-univ-adequate}:}
\begin{proof}
  First note that all the rewrite rules of $\Sigma_U$ are left-linear (no
  rule left hand-side contains two occurrences of the same free
  variable) and there is only one critical pair that is trivially
  joinable: $\bot\,\et\,\bot$ can be rewritten with two different
  rules but the result is the same, namely $\bot$. $\Sigma_U$ is therefore
  weakly orthogonal and thus confluent \cite{oostrom94lfcs}. The
  critical pair could be removed and the system turned into a fully
  orthogonal rewrite system by replacing the rules responsible of the
  critical pair by $\bot\,\et\,\bot\rw\bot$, $\bot\,\et\,\top\rw\bot$,
  $\top\,\et\,\bot\rw\bot$. The weakly orthogonal system is however
  slightly more efficient as some subterms need not be matched upon in
  order to apply a rewrite rule \cite{hondet20fscd}.
  
  Next, note that the rewrite rules form a terminating rewrite system
  as all recursive calls are done on structurally smaller arguments
  \cite{blanqui19fscd}. Therefore, every Dedukti term in $\Sigma_U$
  has a unique normal form. We now need to prove that rewriting
  preserves the semantics but, since rewrite rules are defined on an
  extended algebra, we first need to define the semantics of new
  symbols.

  For all $\b$-normal Dedukti terms $t$ typable in $\Sigma_U$, we
  define its interpretation $\K{t}_\xi$ compositionally by defining,
  for each function symbol $f\!:\!A\in\Sigma_U$, an interpretation
  function $\o{f}\in\o{A}$ where:
  \begin{itemize}
  \item $\o\Type$ is a set big enough to include all the sets below.
    
  \item $\o\T$ is the smallest set
    $X=\{\code\C,\code\B,\code\N,\code\S\}\cup\{\set\,a\mid a\in
    X\}$.
    
  \item $\K{\Pi x:A,B}_\xi=\Pi a\in\K{A}_\xi,\K{B}_{\xi,(x,a)}$ where,
    given a set $A$ and an $A$-indexed family of sets $(B_a)_{a\in
      A}$, $\Pi a\!\in\!\!A,B_a$ is the set of functions $f$ from $A$ to
    $\bigcup\{B_a\mid a\in\!A\}$ such that, for all $a\in\!A$,
    $f(a)\!\in\!B_a$. In particular, $\o{A\to B}$ is the set of
    functions from $\o{A}$ to $\o{B}$.
    
  \item $\o\tau$ is defined as follows:
    \begin{itemize}
    \item $\o\tau\,\code\C=\o\C=\{\lt,\eq,\gt\}$,
    \item $\o\tau\,\code\B=\o\B=\{\top,\bot\}$,
    \item $\o\tau\,\code\N=\o\N$ is the smallest set
      $X=\{\z\}\cup\{\s\,x\mid x\in X\}$,
    \item $\o\tau\,\code\S=\o\S=\{\cst\,G\,i\mid G\in\o\typG,i\in\o\N\}\cup\{\lV\,G\,m\,i\mid G\in\o\typG,m\in\o\N,i\in\o\N\}$,
    \item $\o\tau\,(\set\,a)$ is the smallest set
      $X=\{\nil\,a\}\cup\{\cons\,a\,x\,l\mid x\in\o\tau\,a,l\in X\}$.
    \end{itemize}
  \end{itemize}
  \noindent
  A function symbol $f\in\Sigma_U-\Sigma_I$\footnote{$\Sigma_I$ is
  defined in Fig. \ref{fig-univ}, p. \pageref{fig-univ}.} defined by
  rules $\{l_j\rw r_j\mid j\in J\}$ is interpreted by the unique
  primitive recursive function $\o{f}$ such that, for all valuations
  $\xi$ and $j\in J$, $\K{l_j}_\xi=\K{r_j}_\xi$. This function is
  total and unique because recursive calls are done on structurally
  smaller arguments, rule left-hand sides are non-overlapping (if we
  use the orthogonal version of the system) and cover all possible
  cases. By unfolding the definition of each $\K{r_j}$, this amounts
  to take:
  \begin{itemize}
  \item $\o{\code\C}=\code\C$, $\o\lt=\lt$, $\o\eq=\eq$,
    $\o\gt=\gt$,\\
    $\o\case\,\_\lt\,x\,\_\,\_=x$,
    $\o\case\,\_\eq\,\_\,x\,\_=x$, $\o\case\,\_\gt\,\_\,\_\,x=x$.

  \item $\o{\code\B}=\code\B$, $\o\top=\top$, $\o\bot=\bot$,
    $\top\,\o\et\,\top=\top$,
    $\bot\,\o\et\,x=\bot$, $x\,\o\et\,\bot=\bot$,
    
    $\o\si\,\top\,x\,\_=x$, $\o\si\,\bot\,\_\,x=x$.

  \item $\o{\code\N}=\code\N$, $\o\z=\z$, $\o\s\,x=\s\,x$,
    $\o\cmp_\bN\,\z\,\z=\eq$, $\o\cmp_\bN\,(\s\,x)\,\z=\gt$,
    $\o\cmp_\bN\,\z\,(\s\,y)=\lt$,\\
    $\o\cmp_\bN\,(\s\,x)\,(\s\,y)=\o\cmp_\bN\,x\,y$,
    $x\,\o\le_\bN\,y=\o\case\,\o{\code\B}\,(\o\cmp_\bN\,x\,y)\,\o\top\,\o\top\,\o\bot$.

  \item $\o\set\,a=\set\,a$,
    $\o\nil\,a=\nil\,a$,
    $\o\cons\,a\,x\,l=\cons\,a\,x\,l$,\\
    $\o\fold~a~b~f~(\nil~\_)~y = y$,
    $\o\fold~a~b~f~(\cons~\_~x~S)~y = \o\fold~a~b~f~S~(f~x~y)$,\\
    $\o\add~a~\_~x~(\nil~\_) = \o\cons~a~x~\nil$,\\
    $\o\add~a~f~x~(\o\cons~a~y~S\as T) = \o\case~(\o\set~a)~(f~x~y)~(\o\cons~a~x~T)~T~(\o\cons~a~y~(\o\add~a~f~x~S))$,\\
    $\o\union~\_~\_~(\nil~\_)~T = T$,
    $\o\union~a~f~(\cons~x~S)~T = \o\add~a~f~x~(\o\union~a~f~S~T)$,\\
    $\o\incl~\_~\_~(\nil~\_)~\_ = \o\top$,
    $\o\incl~\_~\_~(\cons~\_~\_~\_)~(\nil~\_) = \o\bot$,\\
    $\o\incl~a~f~(\cons~\_~x~S'\as S)~(\cons~\_~y~T') = \o\case~\code\B~(f~x~y)~\o\bot~(\o\incl~a~f~S'~T')~(\o\incl~a~f~S~T')$,\\
    $\o\cmpset~\_~\_~(\nil~\_)~(\nil~\_) = \o\eq$,\\
    $\o\cmpset~\_~\_~(\nil~\_)~(\cons~\_~\_~\_) = \o\lt$,
    $\o\cmpset~\_~\_~(\cons~\_~\_~\_)~(\nil~\_) = \o\gt$,\\
    $\o\cmpset~a~cmp~(\cons~a~x~S)~(\cons~a~y~T) = \o\case~\o{\code\C}~(cmp~x~y)~\o\lt~(\o\cmpset~a~cmp~S~T)~\o\gt$.

  \item $\o{\code\typG}=\code\typG$, $\o\cmp_\typG=\o\cmpset~\code\N~\o\cmp_\bN$.

  \item $\o{\code\S}=\code\S$, $\o\var~G~m~i=\var~G~m~i$, $\o\cst~G~i=\cst~G~i$,\\
    $\o\cmp_\S~(\cst~\_~\_)~(\var~\_~\_~\_)=\o\lt$,
    $\o\cmp_\S~(\var~\_~\_~\_)~(\cst~\_~\_)=\o\gt$,\\
    $\cmp_\S~(\cst~G~i)~(\cst~H~j) = \o\case~\o{\code\C}~(\o\cmp_\typG~G~H)~\o\lt~(\o\cmp_\bN~i~j)~\o\gt$,\\
    $\o\cmp_\S~(\var~G~x~i)~(\var~H~y~j) = \o\case~\o{\code\C}~(\o\cmp_\typG~G~H)~\o\lt~(\o\case~\o{\code\C}~(\o\cmp_\bN~x~y)~\o\lt~(\o\cmp_\bN~i~j)~\o\gt)~\o\gt$,\\
    $\o\guard~(\var~G~\_~\_)= G$, $\o\guard~(\cst~G~\_)= G$,\\
    $\o\addguard~(\cst~G~i)~s = \o\cst~(\o\union~\o{\code\N}~G~(\o\guard~s))~i$,\\
    $\o\addguard~(\var~G~x~i)~s = \o\var~(\o\union~\o{\code\N}~G~(\o\guard~s))~x~i$,\\
    $\o\incr~(\var~G~x~i) = \o\var~G~x~(\s~i)$,
    $\o\incr~(\cst~G~i) = \o\cst~G~(\s~i)$,\\
    $\var~\_~\_~\_~\o\le~\cst~\_~\_ = \bot$,
    $\cst~G~i~\o\le~\cst~H~j = \o\incl~\o{\code\N}~H~G~\o\et~i~\o\le_\bN~j$,\\
    $\cst~G~i~\o\le~\var~H~\_~j = \o\incl~\o{\code\N}~H~G~\o\et~i~\o\le_\bN~\s~j$,
    $\var~G~x~i~\o\le~\var~H~y~j = \o\incl~\o{\code\N}~H~G~\o\et~x~\o{=}_\bN~y~\o\et~i~\o\le_\bN~j$.
  \end{itemize}
  \noindent
  Hence, for all rules $l\rw r$ defining a symbol
  $f\in\Sigma_U-\Sigma_I$, and for all valuations $\xi$, we have
  $\K{l}_\xi=\K{r}_\xi$. Therefore, by compositionality, for all
  $\b$-normal terms $t$ and $u$ with a type distinct from $\L$, and
  for all valuations $\xi$, if $t\rw u$, then $\K{t}_\xi=\K{u}_\xi$,
  by induction on the depth of the rewriting step. Indeed, if $t_1\rw
  t_2$, then $t_i=C[x/l_i\sigma]$, where $l_1\rw l_2$ is a rewrite
  rule, $\sigma$ is a substitution, and $C$ is a term with a unique
  occurrence of a fresh variable $x$. Since the interpretation is
  defined by induction on the structure of terms, we have
  $\K{t_i}_\xi=\K{C}_{\xi_i}$ where $\xi_i(y)=\K{l_i\sigma}_\xi$ if
  $y=x$, and $\xi(y)$ otherwise. But,
  $\K{l_i\sigma}_\xi=\K{l_i}_{\xi\circ\sigma}$. Therefore, since
  $\K{l_1}_{\xi\circ\sigma}=\K{l_2}_{\xi\circ\sigma}$, we have
  $\K{t_1}_\xi=\K{t_2}_\xi$.

  We now define the interpretation of symbols of $\Sigma_I$ as follows:
  \begin{itemize}
  \item $\o\L=(\bN\to\bN)\to\bN$,
  \item $\o\dzero(\mu)=0$,
  \item $\o\dsucc(x)(\mu)=x(\mu)+1$,
  \item $\o\dmax(x,y)(\mu)=\max(x(\mu),y(\mu))$,
  \item $\o\dimax(x,y)(\mu)=\imax(x(\mu),y(\mu))$,
  \item $\o\dMax(\nil~\_)(\mu)=0$,
  \item
    $\o\dMax(\cons~\_~\w~E)(\mu)=\max(\I{\w^*}_{\mu\circ\i},\o\dMax(E)(\mu))$,
  \end{itemize}
  \noindent
  where $\cdot^*$ translates normal Dedukti terms of type $\L$ back to
  $\o\cL$ as follows:
  \begin{itemize}
  \item $(\dMax\,E)^*=\lMax\,E^*$,
  \item $(\nil\,\_)^*=\emptyset$ and $(\cons\,\_\,x\,l)^*=\{x^*\}\cup l^*$,
  \item $(\cst\,G\,i)^*=\lC\,G^*\,\o{i}$
    and $(\var\,G\,m\,i)^*=\lV\,G^*\,m^*\,\o{i}$,
  \item $m^*=\i^{-1}(\o{m})$,
  \item $\o\z=0$ and $\o{\s\,i}=\o{i}+1$ ($\o\cdot$ is the inverse
    of $\u\cdot$).
  \end{itemize}
  Then, we can prove {\bf(a)} for all $l\in\cL$,
  $\I{l}_\xi=\K{|l|_c}(\xi\circ\i^{-1})$, by induction on $l$.

  We then check that, for all rules $l\rw r$ on the function symbols
  of $\Sigma_I$, and valuations $\xi$, we have $\K{l}_\xi=\K{r}_\xi$. The
  correctness of the rule for $\lsucc$ follows from \cite[Proposition
    42]{geran26csl}. The correctness of the rule for $\lmax$ follows
  from \cite[Proposition 43]{geran26csl}. The correctness of the rule
  for $\limax$ follows from \cite[Proposition 44]{geran26csl} and the
  fact that $\J{w_1}~\o\le~\J{w_2}$ iff $w_1^*\preceq w_2^*$ by
  \cite[Theorem 29]{geran26csl}.

  Hence, by compositionality, we have {\bf(b)} for all $\b$-normal
  terms $t$ and $u$ typable in $\Sigma_U$ and valuations $\xi$, if
  $t\rw u$ then $\K{t}_\xi=\K{u}_\xi$.

  We now prove {\bf(c)} for all closed normal terms $t$ of type $\L$,
  $\K{t}(\mu)=\I{t^*}_{\mu\circ\i}$, by induction on $t$. First,
  $\K{\dMax(\nil\,\_)}(\mu)\!=\!0$ and
  $\I{(\dMax(\nil\,\_))^*}_{\mu\circ\i}\!=\!\I{\lMax\,\emptyset}_{\mu\circ\i}\!=\!0$.
  Second,
  $\K{\dMax(\cons\,\_\,\w\,E)}(\mu)=\max(\I{\w^*}_{\mu\circ\i},\K{\dMax\,E}(\mu))$
  and\\
  $\I{(\dMax(\cons\,\_\,\w\,E))^*}_{\mu\circ\i}=\I{\lMax(\{\w^*\}\cup
    E^*)}_{\mu\circ\i}=\max(\I{\w^*}_{\mu\circ\i},\I{(\lMax\,E)^*}_{\mu\circ\i})$.

  Therefore, we have {\bf(d)} for all $l\in\cL$, if $|l|_c\rw^*t$ and $t$ is
  in normal form, then $l\simeq t^*$ since
  $\I{l}_\xi=\K{|l|_c}(\xi\circ\i^{-1})=\K{t}(\xi\circ\i^{-1})=\I{t^*}_{\xi\circ\i^{-1}\circ\i}=\I{t^*}_\xi$
  by (a), (b) and (c) respectively.

  This allows us to prove the right-to-left direction of the
  theorem. Indeed, assume that $l_1$ and $l_2$ are two terms in $\cL$
  such that $\Sigma_U\th|l_1|_c\equiv|l_2|_c$. Let $t$ be the normal form of
  $|l_1|_c$ and $|l_2|_c$. By (d), we have $l_1\simeq t^*\simeq l_2$.

  We now prove the left-to-right direction.
  
  First note that we have {\bf(e)} if $t$ is in normal form, then
  $t^*$ is a GNF.

  Next, we extend $|\cdot|_c$ on $\o\cL$ as follows:
  \begin{itemize}
  \item $|\lMax\,E|_c=\dMax|E|_\E$,
  \item $|\emptyset|_\E=\nil\,\code\S$,
  \item $|E|_\E=\cons\,\code\S\,|\w|_\S\,|E-\w|_\E$ if $E\neq\emptyset$,
    where $\w$ is the smallest element of $E$ wrt the total ordering
    $\w_1\le\w_2$ such that $\K{\cmp_\S\,|\w_1|_\S\,|\w_2|_\S}=\top$,
  \item $|\lC\,G\,i|_\S=\cst\,|G|_\typG\,\u{i}$,
  \item $|\lV\,G\,\a\,i|_\S=\var\,|G|_\typG\,\u{\i_\a}\,\u{i}$,
  \item $|\emptyset|_\typG=\nil\,\code\N$,
  \item $|G|_\typG=\cons\,\code\N\,\u{m}\,|G-m|_\typG$ if
    $G\neq\emptyset$, where $m$ is the smallest element of $G$ wrt the
    total ordering $\a_1\le\a_2$ such that $\i_{\a_1}\le\i_{\a_2}$.
  \end{itemize}

  Then, we can prove that, for all $\b$-normal Dedukti term in
  $\Sigma_U$, we have $|t^*|_c=t$, by induction on $t$. Therefore, we
  have {\bf(f)} the function $\cdot^*$ is injective.

  Assume now that we have two levels $l_1$ and $l_2$ such that
  $l_1\simeq l_2$. Let $t_i$ be the normal form of $|l_i|_c$. We have
  $|l_i|_c\rw^* t_i$, $l_i\simeq t_i^*$ by (d), and $t_i^*$ a GNF by
  (e). Therefore, $t_1^*=t_2^*$ and, by (f), $t_1=t_2$. Therefore,
  $\Sigma_U\th_D|l_1|_c\equiv|l_2|_c$.
\end{proof}


\newpage
{\bf Proof of Theorem \ref{th-sound}:}
\begin{proof}
  Let (*) be the assumption that Lean has convertibility of types and
  injectivity of $\U$. For the sake of simplicity, we will drop
  $\chi$ and $\Sigma$ in the following, and simply write $t\deq u$
  instead of $\Sigma\th_D t\equiv u$.

  We start by proving a few small facts.
  \begin{enumerate}[label={\bf(\alph*)}]\itemsep+1mm
  \item\label{a} If $\th_D|\G|$ then $|\G|\th_D|l|:\L$ and
    $|\G|\th_D\U|l|:\Type$.

  \item\label{b} If $\th\G$ then $||\U l||_\G\deq\U|l|$.
    Proof: $||\U l||_\G=\e\,|\chi_\G(\U l)|\,(\sort|l|)\rw\U|l|$.

  \item\label{c} If $\G\th A:\U l$ and
    $|\G|\th_D|A|_\G:||\U l||_\G$ then
    $||\U l||_\G\deq\U|\chi_\G(A)|$.

    \vspace*{1mm}
    Proof: By (\ref{b}), $||\U l||_\G\deq\U|l|$. By
    definition of $\chi$, $\G\th A:\U l'$ where
    $l'=\chi_\G(A)$. By (*), $\th\U l\equiv\U l':s$ and
    $l\simeq l'$. By Theorem \ref{th-univ-poly-correct},
    $|l|\deq|l'|$. Therefore, $||\U l||\deq\U|\chi_\G(A)|$.

  \item\label{d} If $\th\G$ and $\th_D|\G|$ then
    $|\G|\th_D||\U l||_\G:\Type$.

    \vspace*{1mm}
    Proof: By definition,
    $||\U l||_\G=\e\,|l'|\,(\sort|l|)$ where
    $l'=\chi_\G(\U l)$. By (decl),
    $|\G|\th_D\e:\Pi\a:\L,\U\a\to\Type$ and
    $|\G|\th_D\sort:\Pi\a:\L,\U(\dsucc\,\a)$. By (\ref{a}),
    $|\G|\th|l|:\L$. By (app), $|\G|\th_D\sort|l|:\U(\dsucc|l|)$.
    By (sort), $\G\th\U l:\U(\lsucc\,l)$ and, by
    definition of $\chi$, $\G\th\U l:\U l'$. Hence, by (*),
    $\G\th\U(\lsucc\,l)\equiv\U l':s$ and $\lsucc\,l\simeq
    l'$. So, by Theorem \ref{th-univ-poly-correct},
    $\dsucc|l|=|\lsucc\,l|\deq|l'|$ and
    $\U(\dsucc|l|)\deq\U|l'|$.
    Now, by (\ref{a}), $|\G|\th_D\U|l'|:\Type$. Therefore, by
    (conv), $|\G|\th_D\sort|l|:\U|l'|$ and, by (app),
    $|\G|\th_D||\U l||_\G:\Type$.
    
  \item\label{e} If $\G\th A:\U l$ and
    $|\G|\th_D|A|_\G:||\U l||_\G$ then $|\G|\th_D||A||_\G:\Type$.

    \vspace*{1mm}
    Proof: By definition, $||A||_\G=\e\,|l'|\,|A|_\G$ where
    $l'=\chi_\G(A)$. By (decl),
    $|\G|\th_D\e:\Pi\a:\L,\U\a\to\Type$. By (\ref{a}),
    $|\G|\th_D|l'|:\L$. By (\ref{c}),
    $||\U l||_\G\deq\U|l'|$. By (\ref{a}),
    $|\G|\th_D\U|l'|:\Type$. Therefore, by (conv),
    $|\G|\th_D|A|_\G:\U|l'|$ and, by (app),
    $|\G|\th_D||A||_\G:\Type$.

  \item\label{f} If $\G\th A\!:\!\U l$,
    $|\G|\th_D|A|_\G\!:\!||\U l||_\G$, $\G,x\!:\!A\th
    B\!:\!\U l'$ and
    $|\G|,x\!:\!||A||_\G\th_D|B|_{\G,x:A}\!:\!||\U l'||_{\G,x:A}$, then
    $||\Pi x\!:\!A,B||_\G\deq\Pi x\!:\!||A||_\G,||B||_{\G,x:A}$ and
    $|\G|\th_D\Pi x\!:\!||A||_\G,||B||_{\G,x:A}\!:\!\Type$.

    \vspace*{1mm}
    Proof: Let $l_A=\chi_\G(A)$, $l_B=\chi_{\G,x:A}(B)$ and
    $l=\chi_\G(\Pi x:A,B)$. By definition, $||\Pi
    x:A,B||_\G=\e\,|l|\,(\pi\,|l_A|\,|l_B|\,|A|_\G\,(\l
    x\!:\!||A||_\G,|B|_{\G,x:A}))\rw \Pi
    x:\e\,|l_A|\,|A|_\G,\e\,|l_B|\,((\l
    x\!:\!||A||_\G,|B|_{\G,x:A})x)$ $\rw_\b \Pi
    x:\e\,|l_A|\,|A|_\G,\e\,|l_B|\,|B|_{\G,x:A}=\Pi
    x\!:\!||A||_\G,||B||_{\G,x:A}$.

    Moreover, by (\ref{e}), $|\G|\th_D||A||_\G:\Type$ and
    $|\G|,x\!:\!||A||_\G\th_D||B||_{\G,x:A}:\Type$. Therefore, by
    (prod), $|\G|\th_D\Pi x\!:\!||A||_\G,||B||_{\G,x:A}:\Type$.

  \item\label{g} If $\G\th t:A$ and $\D\th\sigma:\G$, then
    $|t\sigma|_\D\deq|t|_\G|\sigma|_\D$ where
    $x|\sigma|_\D=|x\sigma|_\D$. Moreover, if $(\G,t)\in\cT$, then
    $||t\sigma||_\D\deq||t||_\G|\sigma|_\D$.

    \vspace*{1mm}
    Proof: First note that, if
    $|t\sigma|_\D\deq|t|_\G|\sigma|_\D$ and $(\G,t)\in\cT$, then
    $||t\sigma||_\D\deq||t||_\G|\sigma|_\D$. Indeed, by definition, we
    have $||t||_\G=\e\,l\,|t|_\G$ where $l=\chi_\G(t)$,
    $||t\sigma||_\D=\e\,l'\,|t\sigma|_\D$ where $l'=\chi_\D(t\sigma)$,
    $\G\th t:\U l$ and $\D\th t\sigma:\U l'$. But, by
    substitution, we also have $\D\th t\sigma:\U l$. Hence, by
    (*), $l\simeq l'$ and, by Theorem \ref{th-univ-poly-correct},
    $|l|\deq|l'|$. Therefore, $||t\sigma||_\D\deq||t||_\G|\sigma|_\D$.
    
    We now prove that, if $\G\th t:A$ and $\D\th\sigma:\G$, then
    $|t\sigma|_\D\deq|t|_\G\theta$ where $\theta=|\sigma|_\D$, by
    induction on $t$:
    \begin{itemize}
    \item $t=\U l$. $|t\sigma|_\D=\U l=|t|_\G\theta$.
    \item $t=x$. $|t\sigma|_\D=|x\sigma|_\D=x\theta=|x|_\G\theta$.
    \item
      $t=uv$. $|(uv)\sigma|_\D=|(u\sigma)(v\sigma)|_\D=|u\sigma|_\D|v\sigma|_\D$. By
      induction hypothesis, $|u\sigma|_\D\deq|u|_\G\theta$ and
      $|v\sigma|_\D\deq|v|_\G\theta$. Thus,
      $|(uv)\sigma|_\D\eq(|u|_\G\theta)(|v|_\G\theta)=(|u|_\G|v|_\G)\theta=|uv|_\G\theta$.
    \item $t=\l x\!:\!A,u$. Wlog we can assume that
      $x\notin\FV(\sigma)=\dom(\sigma)\cup\{\FV(x\sigma)\mid
      x\in\dom(\sigma)\}$. Hence, $|t\sigma|_\D=|\l
      x\!:\!A\sigma,u\sigma|_\D=\l
      x\!:\!||A\sigma||_\D,|u\sigma|_\D$. By induction hypothesis,
      $||A\sigma||_\D\deq||A||_\G\theta$ and
      $|u\sigma|_\D\deq|u|_\G\theta$. Therefore, $|t\sigma|_\D\deq\l
      x\!:\!||A||_\G\theta,|u|_\G\theta=(\l
      x\!:\!||A||_\G,|u|_\G)\theta=|t|_\G\theta$.
    \item $t=\Pi x\!:\!A,B$. We have
      $|t|_\G=\pi\,l_A\,l_B\,|A|_\G\,(\l x\!:\!||A||_\G,|B|_\G)$ where
      $l_A=\chi_\G(A)$ and $l_B=\chi_{\G,x:A}(B)$. Wlog we can assume
      that $x\notin\FV(\sigma)$. Hence, $|t\sigma|_\D=|\Pi
      x\!:\!A\sigma,B\sigma|_\D=\pi\,l_A'\,l_B'\,|A\sigma|_\D\,(\l
      x\!:\!||A\sigma||_\D,|B\sigma|_\D)$ where
      $l_A'=\chi_\D(A\sigma)$ and
      $l_B'=\chi_{\D,x:A\sigma}(B\sigma)$. By induction hypothesis,
      $|A\sigma|_\D\deq|A|_\G\theta$,
      $||A\sigma||_\D\deq||A||_\G\theta$ and
      $|B\sigma|_\D\deq|B|_\G\theta$. By definition of $\chi$, we have
      $\G\th A:\U l_A$, $\G,x\!:\!A\th B:\U l_B$, $\D\th
      A\sigma:\U l_A'$ and $\D,x\!:\!A\sigma\th
      B\sigma:\U l_B'$. But, by subtitution, we also have $\D\th
      A\sigma:\U l_A$ and $\D,x\!:\!A\sigma\th
      B\sigma:\U l_B$. Hence, by (*), $l_A\simeq l_A'$ and
      $l_B\simeq l_B'$ and, by Theorem \ref{th-univ-poly-correct},
      $|l_A|\deq|l_A'$ and $l_B|\deq|l_B'|$. Therefore,
      $|t\sigma|_\D\deq|t|_\G\theta$.
    \end{itemize}

  \item\label{h} If $\G\th A_1\equiv A_2:s\in\cS$ and
    $|A_1|_\G\deq|A_2|_\G$ then $||A_1||_\G\deq||A_2||_\G$.

    \vspace*{1mm}
    Proof: By definition, $s=\U l$ for some $l$,
    $||A_i||=\e\,|\chi_\G(A_i)|\,|A_i|_\G$ and $\G\th
    A_i:\U\chi_\G(A_i)$. By inversion, $\G\th A_i:s$. By (*),
    $\G\th\U\chi_\G(A_i)\equiv s:s'$ for some $s'$, and
    $\chi_\G(A_i)\simeq l$. Hence,
    $\chi_\G(A_1)\simeq\chi_\G(A_2)$. By Theorem
    \ref{th-univ-poly-correct},
    $|\chi_\G(A_1)|\deq|\chi_\G(A_2)|$. Therefore,
    $||A_1||_\G\deq||A_2||_\G$.

  \item\label{i} Let $\G\,\cC_\equiv\,\G'$ iff $\G=\G_1,x:A,\G_2$,
    $\G'=\G_1,x:A',\G_2$ and $\G_1\th A\equiv A':s$. If $\G\th t:A$
    and $\G\,\cC_\equiv\,\G'$ then $\G'\th t:A$.

    \vspace*{1mm}
    Proof by induction on $\G\th t:A$:
    \begin{itemize}
    \item (decl) $\cfrac{\th\G\quad \G\th A:s\in\cS\quad
      x\notin\G}{\th\G,x\!:\!A}$. There are two cases:
      \begin{itemize}
      \item $\G,x\!:\!A\,\cC_\equiv\,\G',x\!:\!A$ with
        $\G\,\cC_\equiv\,\G'$. By induction hypothesis, we have
        $\th\G'$ and $\G'\th A:s$. Therefore, by (decl),
        $\th\G',x\!:\!A$.
      \item $\G,x\!:\!A\,\cC_\equiv\,\G,x\!:\!A'$ with $\G\th A\equiv
        A':s'$ for some $s'$. Therefore, by (decl), $\th\G,x:A'$.
      \end{itemize}
    \item (var) $\cfrac{\th\G\quad x\!:\!A\in\G}{\G\th x:A}$. By
      induction hypothesis, $\th\G'$. Moreover,
      $\G=\G_1,x:A,\G_2$. There are two cases:
      \begin{itemize}
      \item $\G'=\G_1',x:A,\G_2$ and $\G_1\,\cC_\equiv\,\G_1'$, or
        $\G'=\G_1,x:A,\G_2'$ and $\G_2\,\cC_\equiv\,\G_2'$. Then, $x:A\in\G'$
        and, by (var), $\G'\th x:A$.
      \item $\G'=\G_1,x:A',\G_2$ and $\G_1\th A\equiv A':s$ for some
        $s$. By (var), $\G\th x:A'$ and, by (conv), $\G\th x:A$.
      \end{itemize}
    \item The other cases are easily dealt with by induction
      hypothesis.
    \end{itemize}

  \item\label{j} The Lean typing relation is unchanged if we replace
    (app) by\\ (app') $\cfrac{\G\th t:\Pi x\!:\!A,B\quad \G\th
    u:A\quad \G\th B[x/u]:s\in\cS}{\G\th t\ u:B[x/u]}$, and (proj)
    by\\ (proj') $\cfrac{\G\th v:S[\vl]\vt\quad \G\th
      S[\vl]\vt:s\in\cS}{\G\th v.i:B_i\sigma_i}$, where the fact that
    the type of the conclusion is typable by a sort (a property of
    $\th$) is added in premises.

    \vspace*{1mm}
    Proof: Let $\th'$ be the typing relation with (app') instead of (app),
    and (proj') instead of (proj). One can easily check that both
    $\th$ and $\th'$ satisfy:
    \begin{itemize}
    \item Stability by substitution: if $\G\th t:A$ and
      $\D\th\sigma:\G$ then $\D\th t\sigma:A\sigma$.
    \item Correctness of types: if $\G\th t:A$ then $\G\th A:s$ for
      some $s\in\cS$.
    \end{itemize}
    We trivially have ${\th'}\sle{\th}$. Then, one can prove that
    ${\th}\sle{\th'}$ by induction on $\th$. We only detail the
    case of (app). By induction hypothesis, $\G\th't:\Pi x:A,B$ and
    $\G\th'u:A$. By correctness of types, $\G\th'\Pi x:A,B:s$ for some
    $s$. By inversion of typing rules, $\G,x:A\th'B:s'$ for some
    $s'$. By stability by substitution, $\G\th'B[x/u]:s'$. Therefore,
    by (app'), $\G\th'tu:B[x/u]$.
    
  \item\label{k} If $\G\th t:A$, $\G\sle\G'$ and $\th\G'$, then
    $|t|_\G\deq|t|_{\G'}$.

    \vspace*{1mm}
    Proof: Assume that $|t|_\G$ has a subterm of the form $|l|$ at
    position $p$. By definition, there are $\D$ and $B$ such that
    $l=\chi_{\G,\D}(B)$ and $|t|_{\G'}$ has at position $p$ the
    subterm $|l'|$ where $l'=\chi_{\G',D}(B)$. By definition,
    $\G,\D\th B:\U l$ and $\G',\D\th B:\U l'$. Since $\G\sle\G'$
    and $\th\G'$, by weakening, we also have $\G',\D\th B:\U
    l$. Hence, by (*), we have $l\simeq l'$.

  \item\label{l} If $\G\th A:s\in\cS$, $\G\sle\G'$ and $\th\G'$,
    then $||A||_\G\deq||A||_{\G'}$.

    \vspace*{1mm} Proof: By definition, $||A||_\G=\e|l||A|_\G$ with
    $l=\chi_\G(A)$, $\G\th A:\U l$, $||A||_{\G'}=\e|l'||A|_{\G'}$ with
    $l'=\chi_{\G'}(A)$, and $\G'\th A:\U l'$. Since $\G\sle\G'$ and
    $\th\G'$, by weakening, we also have $\G'\th A:\U l$. Hence, by
    (*), we have $l\simeq l'$. Therefore, $||A||_\G\deq||A||_{\G'}$
    since $|A|_\G\deq|A|_{\G'}$ by (\ref{k}).
  \end{enumerate}
  
  \noindent
  We then prove the theorem by induction on $\G\th t:A$ using (app').
  \begin{itemize}
  \item (empty) $\cfrac{}{\th\emptyset}$. We have $\th_D\emptyset$ by
    (empty).

  \item (decl) $\cfrac{\th\G\quad \G\th A:s\in\cS\quad
    x\notin\G}{\th\G,x\!:\!A}$. We have to prove
    $\th_D|\G|,x\!:\!||A||_\G$. By induction hypothesis, $\th_D|\G|$
    and $|\G|\th_D|A|_\G:||s||_\G$. By (\ref{e}),
    $|\G|\th_D||A||_\G:\Type$. Therefore, by (decl),
    $\th_D|\G|,x\!:\!||A||_\G$.

  \item (var) $\cfrac{\th\G\quad x\!:\!A\in\G}{\G\th x:A}$. We have to
    prove $|\G|\th_D|x|_\G:||A||_\G$. By definition, $|x|_\G=x$,
    $\G=\G_1,x\!:\!A,\G_2$ and $x:||A||_{\G_1}\in|\G|$. Therefore, by
    (var), $|\G|\th_D x:||A||_{\G_1}$. Since
    $||A||_{\G_1}\deq||A||_\G$ by (\ref{l}) and
    $|\G|\th_D||A||_\G:\Type$ by (\ref{e}), we have $|\G|\th_D
    x:||A||_\G$ by (conv).

  \item (lam) $\cfrac{\G,x\!:\!A\th t:B\quad \G\th\Pi
    x\!:\!A,B:s\in\cS}{\G\th \l x\!:\!A,t:\Pi x\!:\!A,B}$. We have to
    prove $|\G|\th_D|\l x\!:\!A,t|_\G:||\Pi x\!:\!A,B||_\G$. By
    definition, $|\l x\!:\!A,t|_\G=\l x\!:\!||A||_\G,|t|_\G$. Since
    $\G\th\Pi x\!:\!A,B:s$, we have $\G\th A:\U l_A$ and
    $\G,x\!:\!A\th B:\U l_B$ for some $l_A$ and $l_B$ as
    sub-judgments. Hence, by induction hypothesis, we have
    $|\G|,x\!:\!||A||_\G\th_D|t|_\G:||B||_{\G,x:A}$, $|\G|\th_D|\Pi
    x\!:\!A,B|_\G:||s||_\G$, $|\G|\th_D|A|_\G:||\U l_A||_\G$ and
    $|\G|,x\!:\!||A||_\G\th_D|B|_{\G,x:A}:||\U l_B||_{\G,x:A}$.
    By (\ref{f}), $||\Pi x\!:\!A,B||_\G\deq\Pi
    x\!:\!||A||_\G,||B||_{\G,x:A}$ and $|\G|\th_D\Pi
    x\!:\!||A||_\G,||B||_{\G,x:A}:\Type$. Hence, by (lam),
    $|\G|\th_D\l x\!:\!||A||_\G,|t|_\G:\Pi
    x\!:\!||A||_\G,||B||_{\G,x:A}$. By (\ref{e}), $|\G|\th_D||\Pi
    x\!:\!A,B||_\G:\Type$. Therefore, by (conv), $|\G|\th_D\l
    x\!:\!||A||_\G,t:||\Pi x\!:\!A,B||_\G$.

  \item (app') $\cfrac{\G\th t:\Pi x\!:\!A,B\quad \G\th u:A\quad \G\th
    B[x/u]:s\in\cS}{\G\th t\,u:B[x/u]}$. We have to prove
    $|\G|\th_D|t\,u|_\G:||B[x/u]||_\G$. By definition,
    $|t\,u|=|t|\,|u|$. By induction hypothesis, $|\G|\th_D|t|_\G:||\Pi
    x\!:\!A,B||_\G$, $|\G|\th_D|u|_\G:||B||_\G$ and
    $|\G|\th_D|B[x/u]|_\G:||s||_\G$. By (\ref{f}), $||\Pi
    x\!:\!A,B||_\G\deq\Pi x\!:\!||A||_\G,||B||_{\G,x:A}$ and
    $|\G|\th\Pi x\!:\!||A||_\G,||B||_{\G,x:A}:\Type$. Hence, by
    (conv), $|\G|\th_D|t|_\G:\Pi x\!:\!||A||_\G,||B||_{\G,x:A}$ and,
    by (app), $|\G|\th|t|_\G\,|u|_\G:||B||_{\G,x:A}[x/|u|_\G]$. By
    (\ref{g}), $||B||_{\G,x:A}[x/|u|_\G]\deq||B[x/u]||_\G$. By
    (\ref{e}), $|\G|\th_D||B[x/u]||_\G:\Type$. Therefore, by (conv),
    $|\G|\th_D|t\,u|_\G:||B[x/u]||_\G$.

  \item (sort) $\cfrac{\th\G\quad (s,s')\in\cA}{\G\th s:s'}$. We
    have to prove that $|\G|\th_D|s|_\G:||s'||_\G$. By definition,
    $s=\U l$ for some $l$, $s'=\U(\lsucc\,l)$ and
    $|s|=\sort|l|$. By induction hypothesis, we have $\th_D|\G|$. By
    (decl), $|\G|\th_D\sort:\Pi\a:\L,\U(\lsucc\,\a)$. By (app),
    $\th_D s:\U(\lsucc|l|)$. By (\ref{b}),
    $\U(\lsucc|l|)=\U|\lsucc\,l|\deq||s'||$. By (\ref{d}),
    $|\G|\th_D||s'||_\G:\Type$. Therefore, by (conv),
    $|\G|\th_D|\U l|:||\U(\lsucc\,l)||$.
    
  \item (prod) $\cfrac{\G\th A:s_A\quad \G,x\!:\!A\th B:s_B\quad
    ((s_A,s_B),s)\in\cP}{\G\th\Pi x\!:\!A,B:s}$. We have to prove
    $|\G|\th_D|\Pi x\!:\!A,B|_\G:||s||_\G$.
    By definition, $s_A=\U l_A$, $s_B=\U l_B$,
    $s=\U l$, $l=\limax\,l_A\,l_B$, $|\Pi
    x\!:\!A,B|_\G=\pi\,|l_A'|\,|l_B'|\,|A|_\G\,(\l
    x\!:\!||A||_\G,|B|_{\G,x:A})$, $l_A'=\chi_\G(A)$,
    $l_B'=\chi_{\G,x:A}(B)$. By (decl), $|\G|\th_D\pi:\Pi \a_1:\L,\Pi
    \a_2:\L,\Pi
    t_1:\U\a_1,(\e\,\a_1\,t_1\to\U\a_2)\to\U(\dimax\,\a_1\,\a_2)$.
    By induction hypothesis, $|\G|\th_D|A|_\G:||s_A||_\G$ and
    $|\G|,x\!:\!||A||_\G\th_D|B|_{\G,x:A}:||s_B||_{\G,x:A}$. Moreover:
    \begin{itemize}
    \item By (\ref{a}), $|\G|\th|l_A'|:\L$ and $|\G|\th|l_B'|:\L$.
    \item By (\ref{c}), $||s_A||_\G\deq\U|l_A'|$. By (\ref{a}),
      $|\G|\th_D\U|l_A'|:\Type$. So, by (conv),
      $|\G|\th_D|A|_\G:\U|l_A'|$.
    \item By (\ref{d}), $|\G|\th_D||A||_\G:\Type$ and
      $|\G|,x\!:\!||A||_\G|\th_D||s_B||_{\G,x:A}:\Type$. By (prod),
      $|\G|\th_D||A||_\G\to||s_B||_{\G,x:A}:\Type$. By (lam), $\l
      x\!:\!||A||_\G,|B|_{\G,x:A}:||A||_\G\to||s_B||_{\G,x:A}$. By
      (\ref{c}), $||s_B||_{\G,x:A}\deq\U|l_B'|$. By (\ref{a}),
      $|\G|\th_D\U|l_B'|:\Type$. Therefore, by (conv), $|\G|\th_D\l
      x:||A||_\G,|B|_{\G,x:A}:||A||_\G\to\U|l_B'|$.
    \end{itemize}
    Therefore, by (app), $|\G|\th_D|\Pi
    x\!:\!A,B|_\G:\U(\dimax\,|l_A'|\,|l_B'|)$.
    Since $\G\th A:\U l_A$ and $\G\th A:\U l_A'$, by (*),
    $l_A\simeq l_A'$ and, by Theorem \ref{th-univ-poly-correct},
    $|l_A|\deq|l_A'|$. Since $\G,x\!:\!A\th B:\U l_B$ and
    $\G,x\!:\!A\th B:\U l_B'$, by (*), $l_B\simeq l_B'$ and, by
    Theorem \ref{th-univ-poly-correct}, $|l_B|\deq|l_B'|$.
    Hence,
    $\U(\dimax\,|l_A'|\,|l_B'|)\deq\U(\dimax\,|l_A|\,|l_B|)=\U|l|\deq||s||_\G$
    by (\ref{b}). And since $|\G|\th_D||s||_\G:\Type$ by (\ref{d}),
    we have $|\G|\th_D|\Pi x\!:\!A,B|_\G:||s||_\G$ by (conv).

  \item (proj') $\cfrac{\G\th v:S[\vl]\vt\quad \G\th
    S[\vl]\vt:s\in\cS}{\G\th v.i:B_i\sigma}$. We have to prove that
    $|\G|\th_D|v.i|_\G:||B_i\sigma||_\G$. This follows from the fact
    that $v.i$ can be defined as $\prj_S^i[\vl]\vt v$.

    \hide{
    By definition,
    $|v.i|_\G=\prj_S^i\,(\dinst|\vl|)\,|\vt|_\G\,|v|_\G$. By induction
    hypothesis, $\th_D|\G|$. Hence, by (var),
    $|\G|\th_D\prj_S^i:\Pi\va\!:\!\L,||\Pi\vp\!:\!\vA,S[\va]\vp\to
    B_i\sigma||_\D$ where $\D$ is the context in which $S$ is
    defined. By (\ref{l}) and (conv),
    $|\G|\th_D\prj_S^i:\Pi\va\!:\!\L,||\Pi\vp\!:\!\vA,S[\va]\vp\to
    B_i\sigma||_\G$. By (\ref{f}), $||\Pi\vp\!:\!\vA,S[\va]\vp\to
    B_i\sigma||_\G\deq\Pi p_1\!:\!||A_1||_{\G_0},\dots,\Pi
    p_n\!:\!||A_n||_{\G_{n-1}},||S[\va]\vp||_{\G_n}\to||B_i\sigma||_{\G_n}$
    where $\G_0=\G$ and
    $\G_{i+1}=\G,p_1\!:\!||A_1||_{\G_0},\dots,p_{i+1}\!:\!||A_{i+1}||_{\G_i}$. By
    induction hypothesis,
    $|\G|\th_D|t_i|_\G:||A_i\sigma||_\G$. \TODO{}
    }
    
  \item (conv) $\cfrac{\G\th t:A\quad \G\th A\equiv B:s\in\cS}{\G\th
    t:B}$. We have to prove $|\G|\th_D|t|_\G:||B||_\G$. By inversion,
    we have $\G\th B:s$ as sub-judgment. By induction hypothesis,
    $|\G|\th_D|t|_\G:||A||_\G$ and $|\G|\th_D|B|_\G:||s||_\G$. By
    (\ref{d}), $|\G|\th_D||B||_\G:\Type$. By (conv), it suffices to
    prove $||A||_\G\deq||B||_\G$. By (\ref{h}), it suffices to prove
    $|A|_\G\deq|B|_\G$, which we do hereafter:
    
  \item (refl) $\cfrac{\G\th t:A}{\G\th t\equiv t:A}$. We have to
    prove $|t|_\G\deq|t|_\G$, which follows by (refl).
    
  \item (sym) $\cfrac{\G\th t\equiv u:A}{\G\th u\equiv t:A}$. By
    induction hypothesis, $|t|_\G\deq|u|_\G$. Thus, by (sym),
    $|u|_\G\deq|t|_\G$.
    
  \item (trans) $\cfrac{\G\th t\equiv u\quad \G\th u\equiv v}{\G\th
    t\equiv v}$. By induction hypothesis, $|t|_\G\deq|u|_\G$ and
    $|u|_\G\deq|v|_\G$. Therefore, by (trans), $|t|_\G\deq|v|_\G$.

  \item (cong$\Pi$) $\cfrac{\G\th A\equiv A':s_A\quad \G,x:A\th
    B\equiv B':s_B\quad ((s_A,s_B),s)\in\cP}{\G\th\Pi x\!:\!A,B\equiv
    \Pi x\!:\!A',B':s}$. We have to prove $|\Pi x\!:\!A,B|_\G\deq|\Pi
    x\!:\!A',B'|_\G$. By definition, $s_A=\U l_A$,
    $s_B=\U l_B$, $|\Pi
    x\!:\!A,B|_\G=\pi\,|\chi_\G(A)|\,|\chi_{\G,x:A}(B)|\,|A|_\G\,(\l
    x\!:\!||A||_\G,|B|_{\G,x:A})$, $|\Pi
    x\!:\!A',B'|_\G=\pi\,|\chi_\G(A')|\,|\chi_{\G,x:A'}(B')|$
    $|A'|_\G\,(\l x\!:\!||A'||_\G,|B'|_{\G,x:A'})$. By induction
    hypothesis, $|A|_\G\deq|A'|_\G$ and
    $|B|_{\G,x:A}\deq|B'|_{\G,x:A}$. Hence,
    $||A||_\G\deq||A'||_\G$. Since $\G\th A:\U l_A$ and $\G\th
    A:\U\chi_\G(A)$, by (*) and Theorem
    \ref{th-univ-poly-correct}, $|l_A|\deq|\chi_\G(A)|$. Since
    $\G,x:A\th B:\U l_B$ and $\G,x:A\th
    B:\U\chi_{\G,x:A}(B)$, by (*) and Theorem
    \ref{th-univ-poly-correct}, $|l_B|\deq|\chi_{\G,x:A}(B)|$.  Since
    $\G,x:A\th B:\U l_B$ and $\G\th A\equiv A':s_A$, by (\ref{i}),
    $\G,x:A'\th B:\U l_B$. Since moreover $\G,x:A'\th
    B':\U\chi_{\G,x:A'}(B')$, by (*) and Theorem
    \ref{th-univ-poly-correct},
    $|l_B|\deq|\chi_{\G,x:A'}(B')|$. Hence,
    $|\chi_{\G,x:A}(B)|\deq|\chi_{\G,x:A'}(B')|$ and $|\Pi
    x\!:\!A,B|_\G\deq|\Pi x\!:\!A',B'|_\G$.

  \item (cong$\l$) $\cfrac{\G\th A\equiv A':s_A\in\cS\quad \G,x:A\th
    t\equiv t':B\quad \Pi x:A,B:s_B\in\cS}{\G\th\l x:A,t\equiv \l
    x:A',t':\Pi x:A,B}$. We have to prove $|\l x:A,t|_\G\deq|\l
    x:A',t'|_\G$. By definition, $|\l x:A,t|_\G=\l x:||A||_\G,|t|_\G$
    and $|\l x:A',t'|_\G=\l x:||A'||_\G,|t'|_\G$. By induction
    hypothesis, $|A|_\G\deq|A'|_\G$ and
    $|t|_{\G,x:A}\deq|t'|_{\G,x:A}$. Hence, $||A||_\G\deq||A'||_\G$
    and $|\l x:A,t|_\G\deq|\l x:A',t'|_\G$.
    
  \item (cong@) $\cfrac{\G\th t\equiv t':\Pi x:A,B\quad \G\th u\equiv
    u':A}{\G\th t\,u \equiv t'\,u':B[x/u]}$. We have to prove
    $|t\,u|_\G\deq|t'\,u'|_\G$. By definition,
    $|t\,u|_\G=|t|_\G\,|u|_\G$ and $|t'\,u'|_\G=|t'|_\G\,|u'|_\G$. By
    induction hypothesis, $|t|_\G\deq|t'|_\G$ and
    $|u|_\G\deq|u'|_\G$. Therefore, $|t\,u|_\G\deq|t'\,u'|_\G$.
    
  \item (cong$\equiv$) $\cfrac{\G\th t\equiv u:A\quad \G\th A\equiv
    B:s\in\cS}{\G\th t\equiv u:B}$. We have to prove
    $|t|_\G\deq|u|_\G$, which follows by induction hypothesis.

  \item (cong-proj) $\cfrac{\G\th v\equiv v':S[\vl]\vt}{\G\th
    v.i\equiv v'.i:B_i\sigma_i}$. We have to prove that
    $|v.i|_\G\deq|v'.i|_\G$, which holds since, by induction
    hypothesis, $|v|_\G\deq|v'|_\G$ and, by definition,
    $|v.i|_\G=\prj_S^i\,(\dinst|\vl|)\,|\vt|_\G\,|v|_\G$ and
    $|v'.i|_\G=\prj_S^i\,(\dinst|\vl|)\,|\vt|_\G\,|v'|_\G$.
    
  \item ($\b$-fun) $\cfrac{\G,x:A\th t:B\quad \G\th\Pi x:A,B:s\in\cS\quad
    \G\th u:A}{\G\th {(\l x:A,t)\,u}\equiv t[x/u]:B[x/u]}$. We have to
    prove $|(\l x\!:\!A,t)\ u|_\G\deq|t[x/u]|_\G$. By definition,
    $|(\l x\!:\!A,t)\,u|=(\l
    x\!:\!||A||_\G,|t|_{\G,x:A})$ $|u|_\G\rw_\b|t|_{\G,x:A}[x/|u|_\G]$. Since
    $\G,x\!:\!A\th t:B$ and $\G\th u:A$, by (\ref{g}),
    $|t|_{\G,x:A}[x/|u|_\G]\deq|t[x/u]|_\G$. Therefore, $|(\l
    x\!:\!A,t)\ u|_\G\deq|t[x/u]|_\G$.

  \item ($\b$-rec) and ($\b\text{-}\rec\text{-}S$) Immediate since
    Dedukti implements the same reduction rules for recursors as Lean.

  \item ($\b$-quot)
    $\cfrac{\G\th\lift[l,m]ARBfh(\class[l]AR\,a):B}{\G\th\lift[l,m]ARBfh(\class[l]AR\,a)\equiv
    fa:B}$. Immediate since Dedukti implements the same reduction rule
    for quotients as Lean.

  \item ($\b$-proj)
    $\cfrac{\G\th\mk_S[\vl]\vt\,\vu:S[\vl]\vt}{\G\th(\mk_S[\vl]\vt\,\vu).i\equiv
    u_i:B_i\sigma_i}$. By definition,
    $|(\mk_S[\vl]\vt\,\vu).i|_\G=\prj_S^i$
    $(\dinst|\vl|)\,|\vt|_\G\,(\mk_S\vl|\vt|_\G|\vu|_\G\deq|u_i|_\G$.

  \item ($\eta$-fun) $\cfrac{\G\th t:\Pi x:A,B\quad
    x\notin\FV(t)}{\G\th t\equiv\l x\!:\!A,t\ x:\Pi x\!:\!A,B}$. We
    have to prove that $|t|_\G\deq|\l x\!:\!A,t\ x|_\G$. By
    definition, $|\l x\!:\!A,t\ x|_\G=\l
    x\!:\!||A||_\G,|t|_\G\,x$. Since the translation preserves free
    variables, $x\notin|t|_\G$. Therefore, by $\eta$-reduction,
    $|t|_\G\deq|\l x\!:\!A,t\ x|_\G$.

  \item ($\eta$-proj) $\cfrac{\G\th v:S[\vl]\vt\quad \G\th v.1\equiv
    u_1:B_1\sigma_1\quad\dots\quad\G\th v.n\equiv
    u_n:B_n\sigma_n}{\G\th v\equiv \mk_S[\vl]\vt\,\vu}$. We have to
    prove that $|v|_\G\deq|\mk_S[\vl]\vt\,\vu|_\G$. By definition,
    $|\mk_S[\vl]\vt\,\vu|_\G=\mk_S(\dinst|\vl|)|\vt|_\G\,|\vu|_\G$. For
    all $i$, by induction hypothesis we have $|v.i|_\G\deq|u_i|_\G$
    where, by definition,
    $|v.i|_\G=\prj_S^i\,(\dinst|\vl|)\,|\vt|_\G\,|v|_\G$. Thus,
    $|\mk_S[\vl]\vt\,\vu|_\G\deq\mk_S(\dinst|\vl|)|\vt|_\G(\prj_S^1\,(\dinst|\vl|)\,|\vt|_\G\,|v|_\G)\dots(\prj_S^n\,(\dinst|\vl|)\,|\vt|_\G\,|v|_\G)\deq|v|_\G$
    by the ($\eta$-proj) rule in Dedukti.
  \end{itemize}
\end{proof}

\end{document}